%% file: privacy-social.tex
\documentclass[10pt,conference]{IEEEtran}
\usepackage{epsfig,endnotes}
\usepackage{amsmath,algorithmic,algorithm}
\usepackage{comment}
\usepackage{color}
\usepackage{version}

\includeversion{tech-report}

\newtheorem{mydef}{Definition}
\newtheorem{mythm}{Theorem}
\newtheorem{mycor}{Corollary}

\begin{document}

\title{\Large \bf Preserving Link Privacy in Social Network Based Systems }

\author{
{\rm Prateek Mittal, Charalampos Papamanthou, Dawn Song}\\
University of California, Berkeley \\
\{pmittal,cpap,dawnsong\}@eecs.berkeley.edu \\
} 

\maketitle

\input{abstract}
\input{intro}

\input{related}

\input{background}

\input{protocol}

\input{utility}
\input{privacy}
\input{applications}

\input{conclusion}

\section*{Acknowledgments}
We are very grateful to Satish Rao for helpful discussions on 
graph theory, including insights on utility metrics for 
perturbed graphs, and relating mixing times of original and 
perturbed graphs. We would also like to thank Mario Frank, Ling 
Huang, and Adrian Mettler for helpful discussions on analyzing 
link privacy, as well as their feedback on early versions of this 
work. Our analysis of SybilLimit is based on extending a simulator 
written by Bimal Vishwanath.

\bibliographystyle{acm}
\bibliography{refs}

\input{appendix}

%

\end{document}

%% file: abstract.tex
\begin{abstract}
A growing body of research leverages social network based trust relationships to improve the functionality of the system. However, these systems expose users' trust relationships, which is considered sensitive 
information in today's society, to an adversary. 

In this work, we make the following contributions. First, we 
propose an algorithm that perturbs the structure of a social 
graph in order to provide link privacy, at the cost of slight 
reduction in the utility of the social graph. Second we define 
general metrics for characterizing the utility and privacy of 
perturbed graphs. Third, we evaluate the utility and privacy of 
our proposed algorithm using real world social graphs. 
Finally, we demonstrate the applicability of our perturbation 
algorithm on a broad range of secure systems, including Sybil defenses  
and secure routing.
\end{abstract}

%% file: intro.tex
\section{Introduction}

In recent years, several proposals have been put forward 
that leverage user's social network trust relationships 
to improve system security and privacy. Social networks 
have been used for Sybil defense~\cite{sybilguard,sybillimit,sybilinfer,mohaisen:infocom11,tran:infocom11}, 
secure routing~\cite{whanau:nsdi10,x-vine,sprout}, secure reputation systems~\cite{sumup:nsdi09}, mitigating spam~\cite{ostra:nsdi08}, 
censorship resistance~\cite{kaleidoscope}, and anonymous communication~\cite{drac,nagaraja:pet07}.

A significant barrier to the deployment of these systems is that 
they do not protect the privacy of user's trusted social contacts. 
Information about user's trust relationships is considered sensitive 
in today's society; in fact, existing online social networks such 
as Facebook, Google+ and Linkedin provide explicit mechanisms to 
limit access to this information. A recent study by Dey et al.~\cite{dey:nyu11} 
found that more than 52\% of Facebook users hide their social contact information. 

Most protocols that leverage social networks for system 
security and privacy either explicitly reveal users' trust relationships 
to an adversary~\cite{sybilinfer} or allow the adversary to easily perform traffic 
analysis and infer these trust relationships~\cite{sybillimit}. Thus the design of 
these systems is fundamentally in conflict with the current online 
social network paradigm, and hinders deployment.

In this work, we focus on protecting the privacy of users' 
trusted contacts (edge/link privacy, not vertex privacy) 
while still maintaining the utility of higher level systems and 
applications that 
leverage the social graph. Our key insight in this work is that 
for a large class of security applications that leverage social 
relationships, preserving the exact set of edges in the graph 
is not as important as preserving the graph-theoretic structural differences 
between the honest and dishonest users in the system. 

This insight motivates a paradigm of \emph{structured graph perturbation}, 
in which we introduce noise in the social graph (by deleting real edges 
and introducing fake edges) such that the local structures in the original 
social graph are preserved. We believe that for many applications, 
introducing a high level of noise in such a structured fashion does not 
reduce the overall system utility.  

\subsection{Contributions}
In this work, we make the following contributions. 

\begin{itemize}
\item First, we propose a mechanism based on random walks for 
perturbing the structure of the social graph that provides link privacy at the cost of 
a slight reduction in application utility (Section~\ref{sec:protocol}). 

\item We define a general metric for characterizing the utility of 
perturbed graphs. Our utility definition considers the change in graph 
structure from the perspective of a vertex. We formally relate our 
notion of utility to global properties of social graphs, such as 
mixing times and second largest eigenvalue modulus of graphs, and 
analyze the utility properties of our perturbation mechanism 
using real world social networks (Section~\ref{sec:utility}).

\item We define several metrics for characterizing link privacy, and 
consider prior information that an adversary may have for de-anonymizing 
links. We also formalize the relationship between utility and privacy of 
perturbed graphs, and analyze the privacy properties of our perturbation 
mechanism using real world social networks (Section~\ref{sec:privacy}).

\item Finally, we experimentally demonstrate the real world applicability of our 
perturbation mechanism on a broad range of secure systems, including 
Sybil defenses and  secure routing (Section~\ref{sec:applications}).
 In fact, we find that for Sybil defenses, our techniques are of interest even 
outside the context of link privacy.
\end{itemize}

%% file: related.tex
\section{Related Work}

Work in this space can be broadly classified into two categories: 
(a) mechanisms for protecting the privacy of links between 
labeled vertices, and (b) mechanisms for protecting node/graph privacy 
when vertices are unlabeled. The focus of this work is on protecting 
the privacy of relationships among labeled vertices. 

\subsection{Link privacy between labeled vertices} There are two main 
mechanisms for preserving link privacy between labeled vertices. The 
first approach is to perform clustering of vertices and edges, and 
aggregate them into super vertices (e.g.,~\cite{Hay:2008:RSR:1453856.1453873} and~\cite{Zheleva:2007:PPS:1793474.1793485}). In this way, information about 
corresponding sub-graphs can be anonymized. While these clustering 
approaches permit analysis of some macro-level graph properties, they 
are not suitable for black-box application of existing social network 
based applications, such as Sybil defenses. The second class of approaches 
aim to introduce perturbation in the social graph by adding and deleting 
edges and vertices. Next, we discuss this line of research in more 
detail. 

Hay et al.~\cite{hay:umass07} propose a perturbation algorithm which applies a 
sequence of $k$ edge deletions followed by $k$ random 
edge insertions. Candidates for edge deletion are sampled uniformly 
at random from the space of existing edges in graph $G$, while candidates for 
edge insertion are sampled uniformly at \emph{random} from the space of 
edges not in $G$. The key difference between our perturbation mechanism 
and that of Hay et al. is that we sample edges for insertion based on 
the \emph{structure} of the original graph (as opposed to random selection). 
We will compare our approach with that of Hay et al. in 
Section~\ref{sec:privacy}.

Ying and Wu~\cite{ying:sdm08} study the impact of Hay et al.'s perturbation 
algorithms~\cite{hay:umass07} on the spectral properties of graphs, as well as on link privacy. 
They also propose a new perturbation algorithm that aims to preserve 
the spectral properties of graphs, but do not analyze its privacy 
properties. 

Korolova et al.~\cite{Korolova:2008:LPS:1458082.1458123} show that link privacy of %
the overall social network can be breached even if information about the local 
neighborhood of social network nodes is leaked (for example, via a look-ahead feature for 
friend discovery). 

\subsection{Anonymizing the vertices}
Although the techniques described above reveal the \emph{identity} of the vertices in 
the social graph but add noise to the \emph{relationships} between them, there have been 
various works in the literature that aim at anonymizing the identities of the 
nodes in the social network. This line of research is orthogonal to our goals, but we 
describe them for completeness.

The straightforward approach of just removing the identifiers of the nodes before publishing the social graph does not always guarantee privacy, as shown by Backstrom et. al.~\cite{Backstrom:2007:WAT:1242572.1242598}. To deal with this problem, Liu and Terzi~\cite{Liu:2008:TIA:1376616.1376629} propose a systematic framework for identity anonymization on graphs, where they introduce the notion of $k$-degree anonymity. Their goal is to minimally modify the graph by changing the degrees of specially-chosen nodes so that the identity of each individual involved is protected. An efficient version of their algorithm was recently implemented by Lu et al.~\cite{fast-identity}.

Another notion of graph anonymity in social networks is presented by Pei and Zhou~\cite{DBLP:conf/icde/ZhouP08}: A graph is $k$-anonymous if for every node there exist at least $k-1$ other nodes that share isomorphic neighborhoods. This is a stronger definition than the one in~\cite{Liu:2008:TIA:1376616.1376629}, where only vertex degrees are considered. 

Zhou and Pei~\cite{DBLP:journals/kais/0002P11} recently introduced another notion called $l$-diversity for social network anonymization. In this case, each vertex
is associated with some non-sensitive attributes and some
sensitive attributes. Maintaining the privacy of the individual in this scenario is based on the adversary not being able (with high probability) to re-identify the sensitive attribute values of the individual.

Finally, Narayanan and Shmatikov~\cite{DBLP:conf/sp/NarayananS09} show some of the weaknesses of the above anonymization techniques,propose a generic way for modeling the release of anonymized social networks and report on successful de-anonymization attacks on popular networks such as Flickr and Twitter.

\subsection{Differential privacy and social networks}
Sala et al.~\cite{pigmalion} use differential privacy (a more elaborate tool of adding noise) to publish 
social networks with privacy guarantees. Given a social network and a desired level of differential privacy 
guarantee, they extract a detailed structure into degree correlation statistics, introduce noise into the 
resulting dataset, and generate a new synthetic social network with differential privacy. However, their 
approach does not preserve utility of the social graph from the perspective of a vertex (since vertices in 
their graph are unlabeled), and thus cannot be used for many real world applications such as Sybil defenses. %

Also, Rastogi et al.~\cite{Rastogi:2009:RPO:1559795.1559812} introduce a relaxed notion of differential privacy 
for data with relationships so that more expressive queries (e.g., joins) can be supported without hurting utility very much.

\subsection{Link privacy preserving applications}

X-Vine~\cite{x-vine} proposes to perform DHT routing using social links, in 
a manner that preserves the privacy of social links. However, the threat model 
in X-Vine excludes adversaries that have prior information about the social graph. 
Thus in real world settings, X-Vine is vulnerable to the Narayanan-Shmatikov 
attack~\cite{DBLP:conf/sp/NarayananS09}. Moreover the techniques in X-Vine are 
specific to DHT routing, and cannot be used to design a general purpose 
defense mechanism for social network based applications, which is the focus 
of this work.

%% file: background.tex
\section{Basic Theory}

Before we introduce our perturbation mechanism, we present some 
notation and background on graph theory needed to understand the 
paper. 

Let us denote the social graph as $G=(V,E)$, comprising the set of 
vertices $V$ (wlog assume the vertices have labels $1,\ldots,n$), and the set of edges $E$, where $|V|=n$ and $|E|=m$. The focus of this paper 
is on undirected graphs, where the edges are symmetric. Let $A_G$ denote the 
$n\times n$ \emph{adjacency matrix} corresponding to the graph $G$, namely if $(i,j) \in E$, then $A_{ij}=1$, otherwise $A_{ij}=0$. 

A \emph{random walk} on a graph $G$ starting at a vertex $v$ is a 
sequence of vertices comprising a random neighbor $v_1$ of $v$, 
then a random neighbor $v_2$ of $v_1$ and so on. A random 
walk on a graph can be viewed as a Markov chain. We denote 
the \emph{transition probability matrix} of the 
random walk/Markov chain as $P$, given by: 

\begin{equation}
P_{ij}  = 
\begin{cases}
\frac{1}{\mathsf{deg}(i)} & \text{ if }(i,j)\text{ is an edge in $G$}\,, \\
0 &  \text{  otherwise}\,.
\end{cases}
\end{equation}
where $\mathsf{deg}(i)$ denotes the degree of the vertex $i$. At any given iteration 
$t$ of the random walk, let us denote with $\pi(t)$ the probability 
distribution of the random walk state at that iteration ($\pi(t)$ is a vector of $n$ entries). The state 
distribution after $t$ iterations is given by $\pi(t) = \pi(0) \cdot P^t$, 
where $\pi(0)$ is the initial state distribution. The probability of a 
$t$-hop random walk starting from $i$ and ending at $j$ is given 
by $P_{ij}^t$.  

For irreducible and aperiodic graphs (which undirected and connected 
\emph{social graphs} are), the corresponding Markov chain is ergodic, and
 the state distribution of the random walk $\pi(t)$ converges to 
a unique stationary distribution denoted by $\pi$. The stationary 
distribution $\pi$ satisfies $\pi = \pi \cdot P$.

For undirected and connected social graphs, we can see that the 
probability distribution $\pi_i = \frac{\mathsf{deg}(i)}{2 \cdot m}$ 
satisfies the equation $\pi=\pi \cdot P$, and is thus the unique 
stationary distribution of the random walk.

Let us denote the eigenvalues of $A$ as $\lambda_1 \geq \lambda_2 \geq \ldots \geq\lambda_{n}$, 
and the eigenvalues of $P$ as $\nu_1 \geq \nu_2 \geq \ldots \geq \nu_{n}$. The eigenvalues of 
both $A$ and $P$ are real. We denote the second largest eigenvalue modulus 
(SLEM) of the transition matrix $P$ as $\mu = \max(|\nu_2|,|\nu_{n}|)$. The eigenvalues of 
matrices A and P are closely related to structural properties of graphs, and are 
considered utility metrics in the literature.

%% file: protocol.tex
\section{Structured Perturbation}
\label{sec:protocol}

\subsection{System Model and Goals}

For the deployment of secure applications that leverage user's trust relationships, we 
envision a scenario where these applications bootstrap user's trust relationships using 
existing online social networks such as Facebook or Google+.  

However, most applications that leverage this information do not make any attempt 
to hide it; thus an adversary can exploit protocol messages to learn the entire 
social graph. 

Our vision is that OSNs can support these applications while protecting the link 
privacy of users by introducing noise in the social graph. Of course the addition of 
noise must be done in a manner that still preserves application utility. Moreover the 
mechanism for introducing noise should be computationally efficient, and must not present 
undue burden to the OSN operator.

We need a mechanism that takes the social graph $G$ as an input, and produces a transformed 
graph $G' = (V,E')$, such that the vertices in $G'$ remain the same as the original 
input graph $G$, but the set of edges is perturbed to protect link privacy. The constraint 
on the mechanism is that application utility of systems that leverage the perturbed 
graph should be preserved. Conventional metrics of utility include degree sequence and 
graph eigenvalues; we will shortly define a general metric for utility of perturbed graphs 
in the following section. There is a tradeoff between privacy of links in the social 
graph and the utility derived out of perturbed graphs. As more and more noise is added 
to the social graph, the link privacy increases, but the corresponding utility decreases.

\subsection{Perturbation Algorithm}

Let $t$ be the parameter that governs how much noise we wish to inject in the social graph. 
We propose that for each node $u$ in graph $G$, we perturb \emph{all} of $u$'s contacts as 
follows. Suppose that node $v$ is a social contact of node $u$. Then we perform a random 
walk of length $t-1$ starting from node $v$. Let node $z$ denote the terminus point 
of the random walk. Our main idea is that instead of the edge $(u,v)$, we will introduce the 
edge $(u,z)$ in the graph $G'$. It is possible that the random walk terminates at either node $u$ 
itself, or that node $z$ is already a social contact of $u$ in the transformed graph $G'$ (due to 
a previously added edge). To avoid self loops and duplicate edges, we perform another random walk
from vertex $v$ until a suitable terminus vertex is found, or we reach a threshold number of tries, 
denoted by parameter $M$. For undirected graphs, the algorithm described so far would double the number 
of edges in the perturbed graphs: for each edge $(u,v)$ in the original graph, an edge would be added 
between a vertex $u$ and the terminus point of the random walk from vertex $v$, as well as between 
vertex $v$ and the terminus point of the random walk from vertex $u$. To preserve the degree distribution, 
we could add an edge between vertex $u$ and vertex $z$ in the transformed graph with probability $0.5$.
However this could lead to low degree nodes becoming disconnected from the social graph with non-trivial 
probability. To account for this case, we add the first edge corresponding to the vertex $u$ with probability 
$1$, while the remaining edges are accepted with a reduced probability to preserve the degree distribution. 
The overall algorithm is depicted in Algorithm~\ref{alg:transform}. The computational complexity of our 
algorithm is $O(m)$.

\begin{algorithm}
\caption{$\textsf{Transform}(G,t,M)$: Perturb undirected graph $G$ using perturbation $t$ and maximum loop count $M$.}
\label{alg:transform}
\begin{algorithmic}
\STATE $G' = \mathsf{null}$;
\STATE  {\bf foreach} vertex $u$ in $G$
\STATE \hspace{0.05in} let $count=1$; 
\STATE \hspace{0.05in} {\bf foreach} neighbor $v$ of vertex $u$
\STATE \hspace{0.1in} let $loop=1$;
\STATE \hspace{0.1in} {\bf do}
\STATE \hspace{0.15in} perform $t-1$ hop random walk from vertex $v$;
\STATE \hspace{0.15in} let $z$ denote the terminal vertex of the random walk;
\STATE \hspace{0.15in} $loop++$;
\STATE \hspace{0.1in} {\bf until} ($u = z$ $\vee$  $(u,z)\in G'$) $\wedge$  ($loop \leq M$)
\STATE \hspace{0.1in} {\bf if} $loop \le M$
\STATE \hspace{0.15in} {\bf if} $count = 1$
\STATE \hspace{0.2in} add edge $(u,z)$ in $G'$
\STATE \hspace{0.15in} {\bf else}
\STATE \hspace{0.2in} let $\mathsf{deg}(u)$ denote degree of $u$ in $G$;
\STATE \hspace{0.2in} add edge $(u,z)$ in $G'$ with probability $\frac{0.5\times \mathsf{deg}(u)-1}{\mathsf{deg}(u)-1}$;
\STATE \hspace{0.05in} $count++$;
\STATE return $G'$;
\end{algorithmic}
\end{algorithm}

\begin{figure*}[ht]
\centering
\mbox{
\hspace{-0.2in}
\hspace{-0.12in}
\begin{tabular}{c}
\psfig{figure=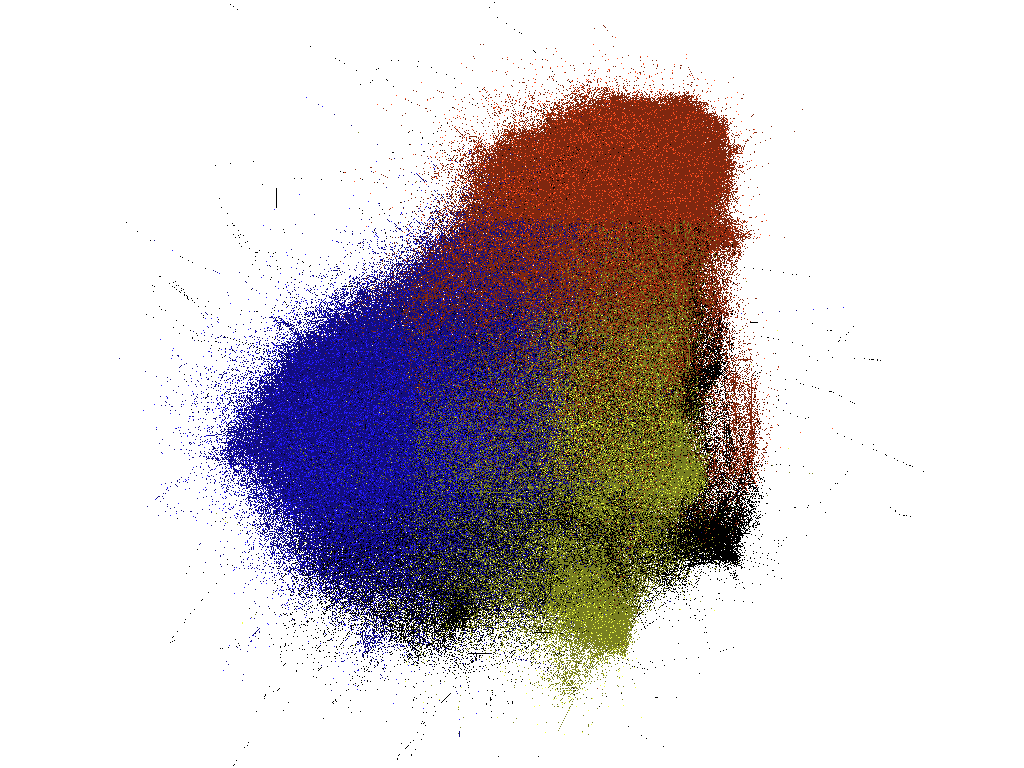,width=0.33 \textwidth}\vspace{-0.00in}\\
{(a)}
\end{tabular}
\begin{tabular}{c}
\psfig{figure=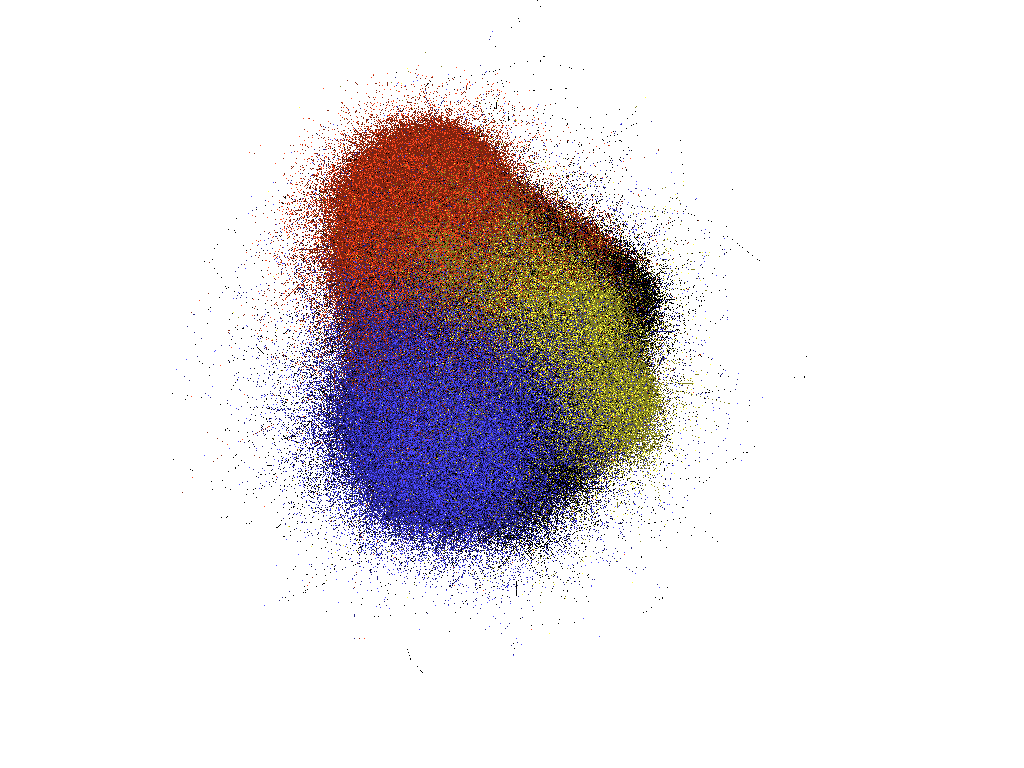,width=0.33 \textwidth}\vspace{-0.00in}\\
{(b)}
\end{tabular}
\hspace{-0.2in}
\begin{tabular}{c}
\psfig{figure=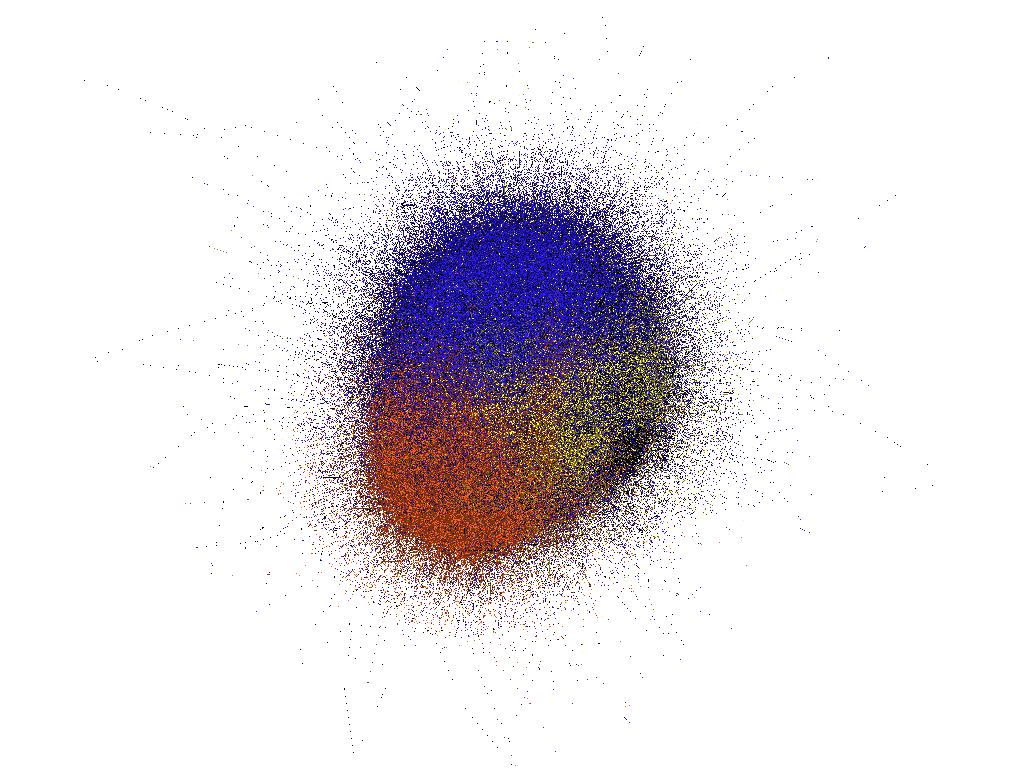,width=0.33 \textwidth}\vspace{-0.00in}\\
{(c)}
\end{tabular}
}
\mbox{
\hspace{-0.2in}
\hspace{-0.12in}
\begin{tabular}{c}
\psfig{figure=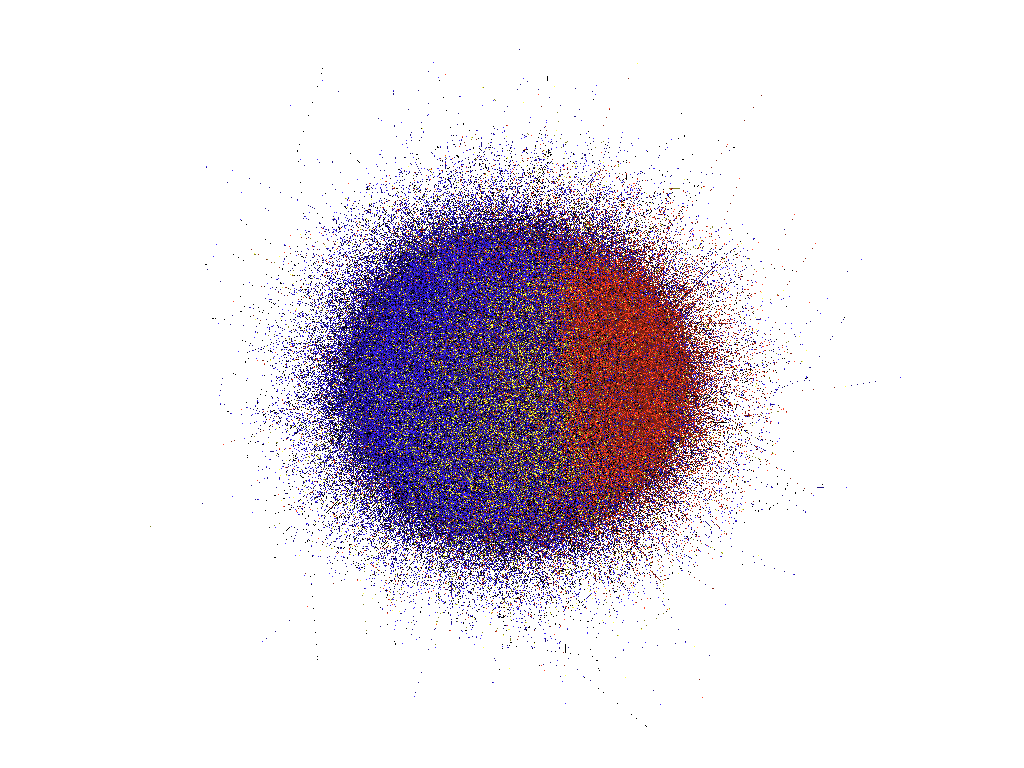,width=0.33 \textwidth}\vspace{-0.00in}\\
{(d)}
\end{tabular}
\begin{tabular}{c}
\psfig{figure=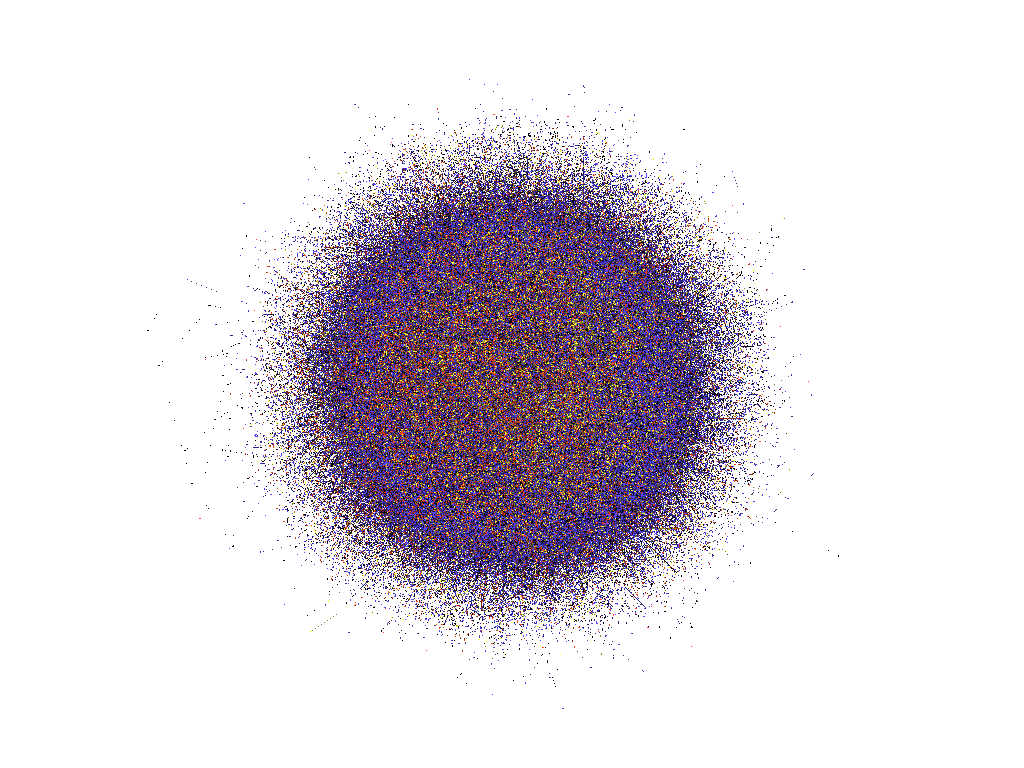,width=0.33 \textwidth}\vspace{-0.00in}\\
{(e)}
\end{tabular}
}
\caption{Facebook dataset link topology (a) Original graph (b) Perturbed, t=5 (c) t=10 (d) t=15, and (e) t=20. 
The color coding in (a) is derived using a modularity based community detection algorithm. For the remaining figures, the color 
coding of vertices is same as in (a). We can see that short random walks preserve the community structure of the social graph, while introducing a significant amount of noise. }
\label{fig:facebook-links}
\end{figure*}

\begin{figure*}[ht]
\centering
\mbox{
\hspace{-0.2in}
\hspace{-0.12in}
\begin{tabular}{c}
\psfig{figure=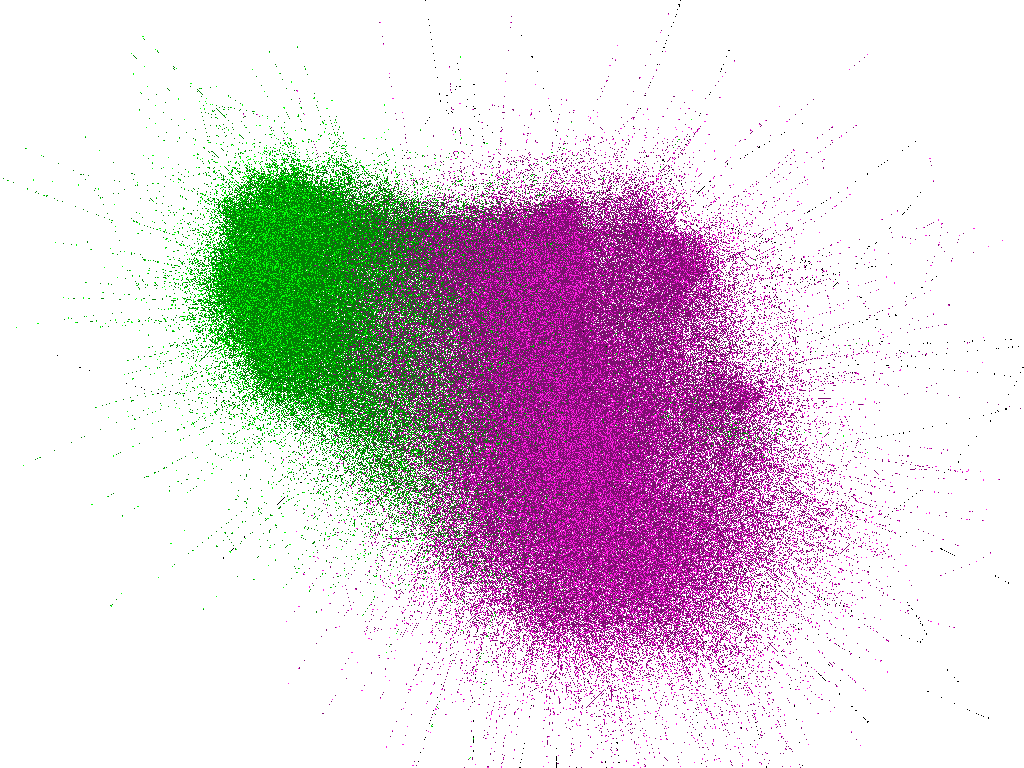,width=0.33 \textwidth}\vspace{-0.00in}\\
{(a)}
\end{tabular}
\begin{tabular}{c}
\psfig{figure=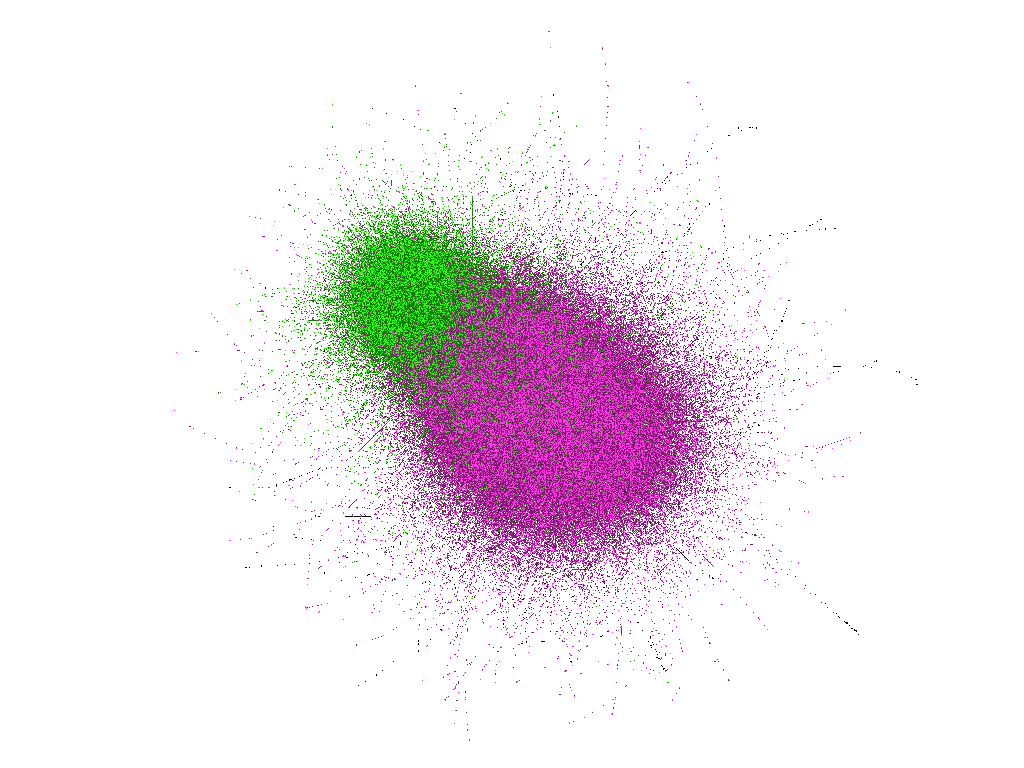,width=0.33 \textwidth}\vspace{-0.00in}\\
{(b)}
\end{tabular}
\hspace{-0.2in}
\begin{tabular}{c}
\psfig{figure=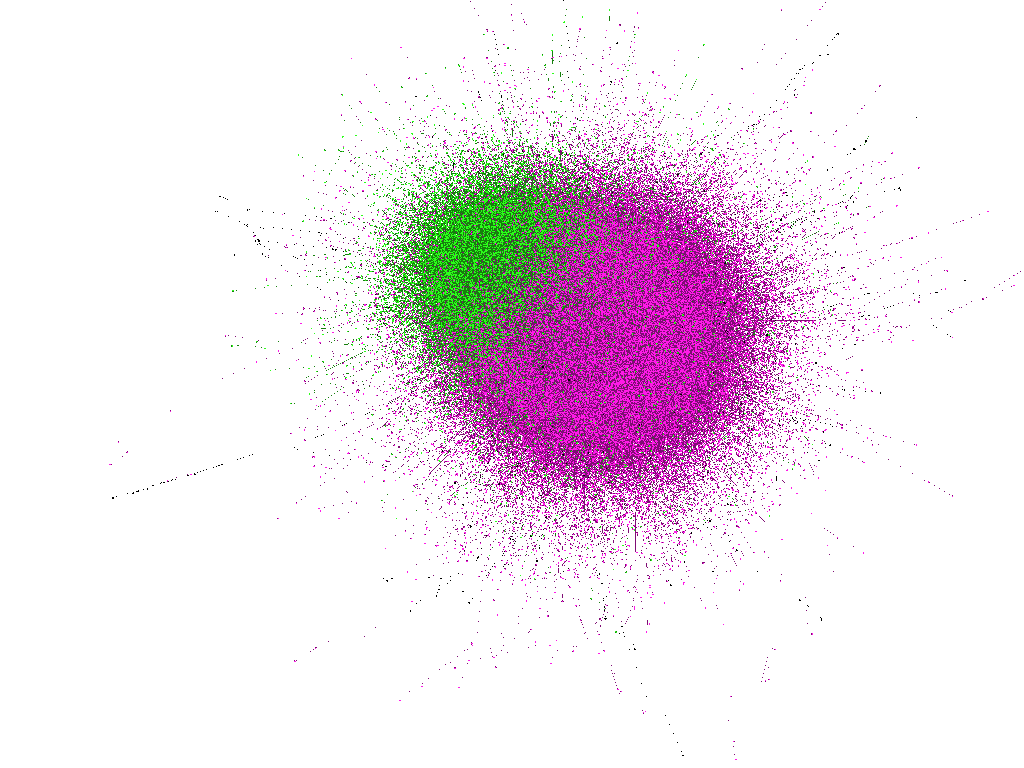,width=0.33 \textwidth}\vspace{-0.00in}\\
{(c)}
\end{tabular}
}
\mbox{
\hspace{-0.2in}
\hspace{-0.12in}
\begin{tabular}{c}
\psfig{figure=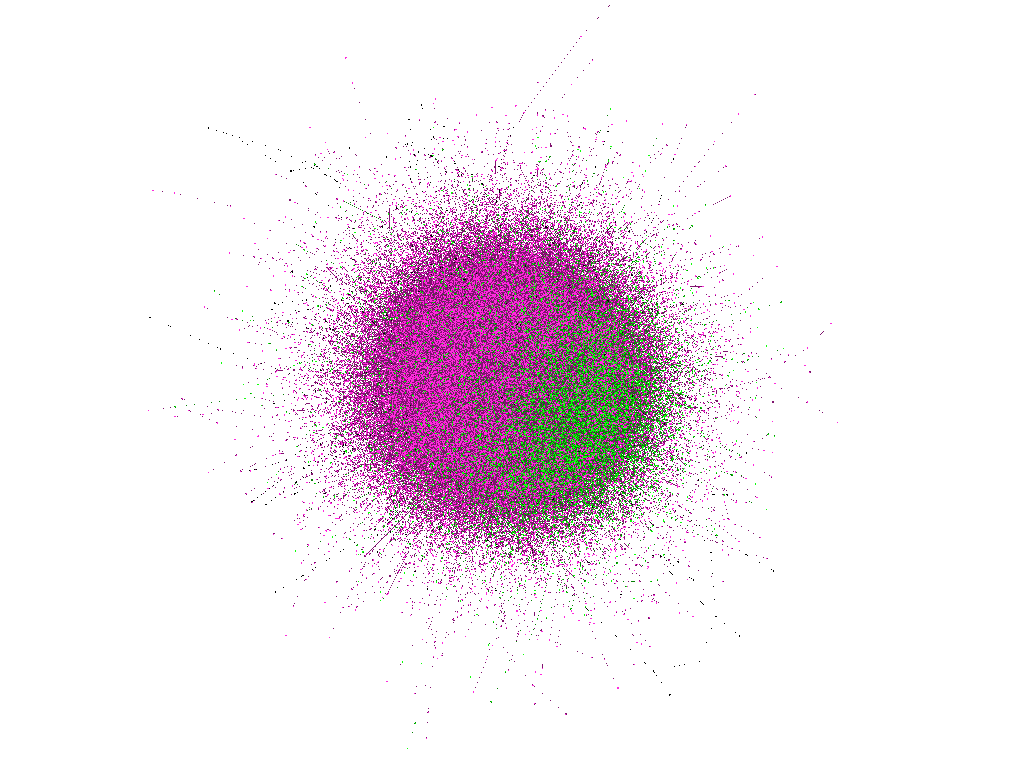,width=0.33 \textwidth}\vspace{-0.00in}\\
{(d)}
\end{tabular}
\begin{tabular}{c}
\psfig{figure=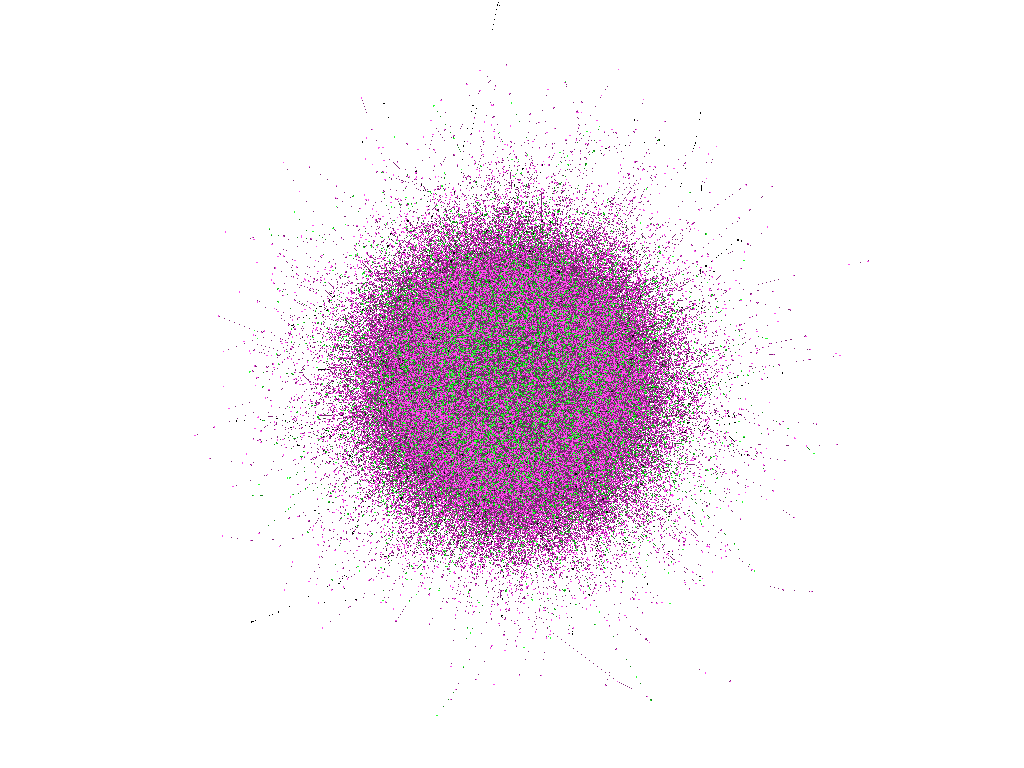,width=0.33 \textwidth}\vspace{-0.00in}\\
{(e)}
\end{tabular}
}
\caption{Facebook dataset interaction topology (a) Original graph (b) Perturbed t=5 (c) t=10 (d) t=15, and (e) t=20. 
We can see that short random walks preserve the community structure of the social graph, while introducing a significant amount of noise. }
\label{fig:facebook-wall}
\end{figure*}

\subsection{Visual depiction of algorithm}

For our evaluation, we consider two real world social network 
topologies (a) \emph{Facebook friendship graph from the New Orleans 
regional network~\cite{vishwanath-wosn09}}: the dataset comprises 63,392 users that have 
816,886 edges amongst them, and (b) \emph{Facebook interaction graph from 
the New Orleans regional network~\cite{vishwanath-wosn09}}: the dataset comprises 43,953 users 
that have 182,384 edges amongst them. Mohaisen et al.~\cite{mohaisen:imc10} found that pre-processing 
social graphs to exclude low degree nodes significantly changes the graph theoretic 
characteristics. Therefore, we did not pre-process the datasets in any way.

Figure~\ref{fig:facebook-links} depicts the original Facebook friendship graph,
and the perturbed graphs generated by our algorithm for varying perturbation parameters, 
using a force directed algorithm for depicting the graph. The color coding of nodes 
in the figure was obtained by running a modularity based community detection 
algorithm on the \emph{original} Facebook friendship graph, which yielded three 
communities. For the perturbed graphs, we used the same color for the vertices as 
in the original graph.  This representation allows us to visually see the perturbation 
in the community structure of the social graph. We can see that for small values of  
the perturbation parameter, the community structure (related to utility) is strongly preserved, even 
though the edges between vertices are randomized. As the perturbation parameter is 
increased, the graph looses its community structure, and eventually begins to 
resemble a random graph.

Figure~\ref{fig:facebook-wall} depicts a similar visualization for the Facebook 
interaction graph. In this setting, we found two communities using a modularity 
based community detection algorithm in the original graph. We can see a similar trend 
in the Facebook interaction graph as well: for small values of perturbation algorithm, 
the community structure is somewhat preserved, even though significant randomization 
has been introduced in the links. In the following sections, we formally quantify the 
utility and privacy properties of our perturbation mechanism.

%% file: utility.tex
\section{Utility}
\label{sec:utility}

In this section, we develop formal metrics to characterize the utility of perturbed graphs, 
and then analyze the utility of our perturbation algorithm.  

\subsection{Metrics}

One approach to measure utility would be to consider global graph theoretic metrics, such as 
the second largest eigenvalue modulus of the graph transition matrix P. However, from a user 
perspective, it may be the case that the users' position in the perturbed graph relative to 
malicious users is much worse, even though the global graph properties remain the same. This 
motivates our first definition of the utility of a perturbed graph from the perspective of 
a single user. 

\begin{mydef}
The vertex utility of a perturbed graph G' for a vertex $v$, with respect to the original graph G, and an 
application parameter $l$ is defined as the statistical distance between the probability distributions 
induced by $l$ hop random walks starting from vertex $v$ in graphs G and G'. 
\end{mydef}

\begin{align}
VU(v,G,G',l) = Distance(P_v^l(G), P_v^l(G'))
\end{align}

$P_v^l$ denotes the $v'th$ row of the matrix $P^l$. The parameter $l$ is linked to higher level 
applications that leverage social graphs. For example, 
Sybil defense mechanisms exploit large scale community structure of social networks, where the 
application parameter $l \geq 10$. For other applications such as recommendation systems, it may be more 
important to preserve the local community characteristics, where $l$ could be set to a smaller value. 

Random walks are intimately linked to the structure of communities and graphs, so it is natural to 
consider their use when defining utility of perturbed graphs. In fact, a lot of security applications 
directly exploit the trust relationships in social graphs by performing random walks themselves, such 
as Sybil defenses and anonymous communication.   

There are several ways to define statistical distance between probability distributions~\cite{statistical-distance}. 
In this work, we consider the following three notions. The total variation distance between two probability 
distributions is a measure of the maximum  difference between the probability distributions for any individual 
element.   

\begin{align}
\textrm{Variation Distance}(P,Q) = ||P-Q||_{tvd} = \sup_{i}|p_i-q_i|
\end{align} 

As we will discuss shortly, the total variation distance is closely related to the computation of several 
graph properties such as mixing time and second largest eigenvalue modulus. However, the total variation distance
only considers the maximum difference between probability distributions corresponding to an element, and 
not the differences in probabilities corresponding to other elements of the distribution. This motivates 
the use of Hellinger distance, which is defined as:

\begin{align}
\textrm{Hellinger Distance}(P,Q) = \nonumber \\ \frac{1}{\sqrt{2}} \cdot \sqrt(\sum_{i=1}^{n} (\sqrt(p_i)-\sqrt(q_i))^2)
\end{align} 

The Hellinger distance is related to the Euclidean distance between the square root vectors of P and Q. Finally, 
we also consider the Jenson-Shannon distance measure, which takes an information theoretic approach of averaging 
the Kullback-Leibler divergence between P and Q, and between Q and P (since Kullback-Leibler divergence by itself is not 
symmetric). 

\begin{align}
\textrm{Jenson-Shannon Distance}(P,Q) = \nonumber \\ \frac{1}{2} \cdot \sum_{i=1}^{n} p_i \log(\frac{p_i}{q_i}) + \frac{1}{2} \cdot \sum_{i=1}^{n} q_i \log(\frac{q_i}{p_i}) 
\end{align} 

Using these notions, we can compute the utility of the perturbed graph with respect to an 
individual vertex (vertex utility). Note that a lower value of $VU(v,G,G',l)$ corresponds to 
higher utility (we want distance between probability distributions over original graph and 
perturbed graph to be low). Using the concept of vertex utility, we can define metrics for 
overall utility of a perturbed graph. 

\begin{mydef}
The overall mean vertex utility of a perturbed graph G' with respect to the original graph G, and an 
application parameter $l$ is defined as the mean utility for all vertices in 
G. Similarly the max vertex utility (worst case) of a perturbed graph G' is defined by computing the maximum 
of the utility values over all vertices in G. 
\end{mydef}

\begin{align}
VU_{mean}(G,G',l) = \sum_{v \in V} \frac{Distance(P_v^l(G), P_v^l(G'))}{|V|}
\end{align}

\begin{align}
VU_{max}(G,G',l) = \max_{v \in V} Distance(P_v^l(G), P_v^l(G'))
\end{align}

The notion of max vertex utility is particularly interesting, specially in conjunction with the 
use of total variation distance. This is because of its relationship to global graph metrics such
as mixing times and second largest eigenvalue modulus, which we demonstrate next. Our analysis shows 
the generality of our formal definition for utility. 

\subsection{Metrics Analysis}

Towards this end, we first introduce the notion of mixing time of a Markov process. The mixing time of a 
Markov process is a measure of the minimum number of steps needed to converge to its unique stationary 
distribution.  Formally, the mixing time of a graph G is defined as:

\begin{align}
\tau_G(\epsilon) = \max_v \min(t | P_v^t(G) - \pi| < \epsilon)
\end{align}

The following two theorems illustrate the bound on global properties of the perturbed graph, using the global properties of 
the original graph, and the utility metric. To improve readability of the paper, we defer the proofs of these theorems to the 
Appendix.

\begin{mythm}
Let us denote the max (worst case) vertex utility distance between the perturbed graph G' and the 
original graph G by $VU_{max}(G,G',l)$, computed as $VU_{max}(G,G',l) = \max_{v \in V} VU(v,G',G,l))$. 
Then we have that $\tau_{G'}(\epsilon+VU_{max}(G,G',\tau_G(\epsilon)) \leq \tau_G(\epsilon)$. 
\label{thm:utility-mixing}
\end{mythm}

Theorem~\ref{thm:utility-mixing} relates the mixing time of the perturbed graph using the mixing time of the original 
graph, and the max vertex utility metric, for application parameter $l=\tau_G(\epsilon)$.

\begin{mythm}
Let us denote the second largest eigenvalue modulus (SLEM) of transition matrix $P_G$ of graph $G$ as $\mu_G$. We can bound the 
SLEM of a perturbed graph G' using the mixing time of the original graph, and the worst case vertex utility 
distance between the graphs as follows: 
\begin{align}
1-\frac{\log n + \log(\frac{1}{\epsilon+VU_{max}(G,G',\tau_G(\epsilon)})}{\tau_G(\epsilon)}\leq \mu_{G'} \leq  \nonumber \\ 
\frac{2 \tau_G(\epsilon)}{2 \tau_G(\epsilon) + \log(\frac{1}{2\epsilon + 2VU_{max}(G,G',\tau_G(\epsilon)}) } \nonumber
\end{align}
\label{thm:utility-slem}
\end{mythm}

Theorem~\ref{thm:utility-slem} relates the second largest eigenvalue modulus of the perturbed graph, using the mixing time of 
the original graph, and the worst case vertex utility metric for application parameter $l=\tau_G(\epsilon)$.

These theorems show the generality of our utility definitions. Mechanisms that provide good utility (have 
low values of $VU_{max}$), introduce only a small change in the mixing time and SLEM of perturbed graphs. 

\begin{figure*}[htp]
\centering
\mbox{
\hspace{-0.2in}
\hspace{-0.12in}
\begin{tabular}{c}
\psfig{figure=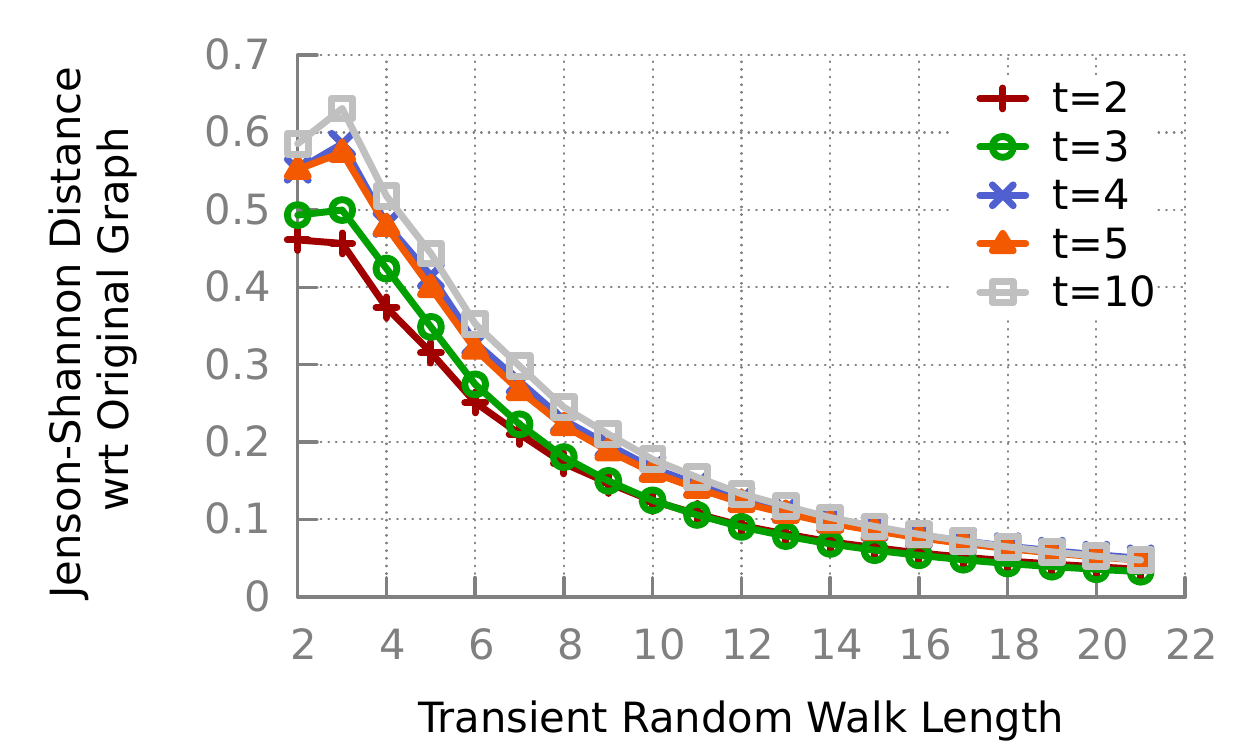,width=0.33 \textwidth}\vspace{-0.00in}\\
{(a)}
\end{tabular}
\begin{tabular}{c}
\psfig{figure=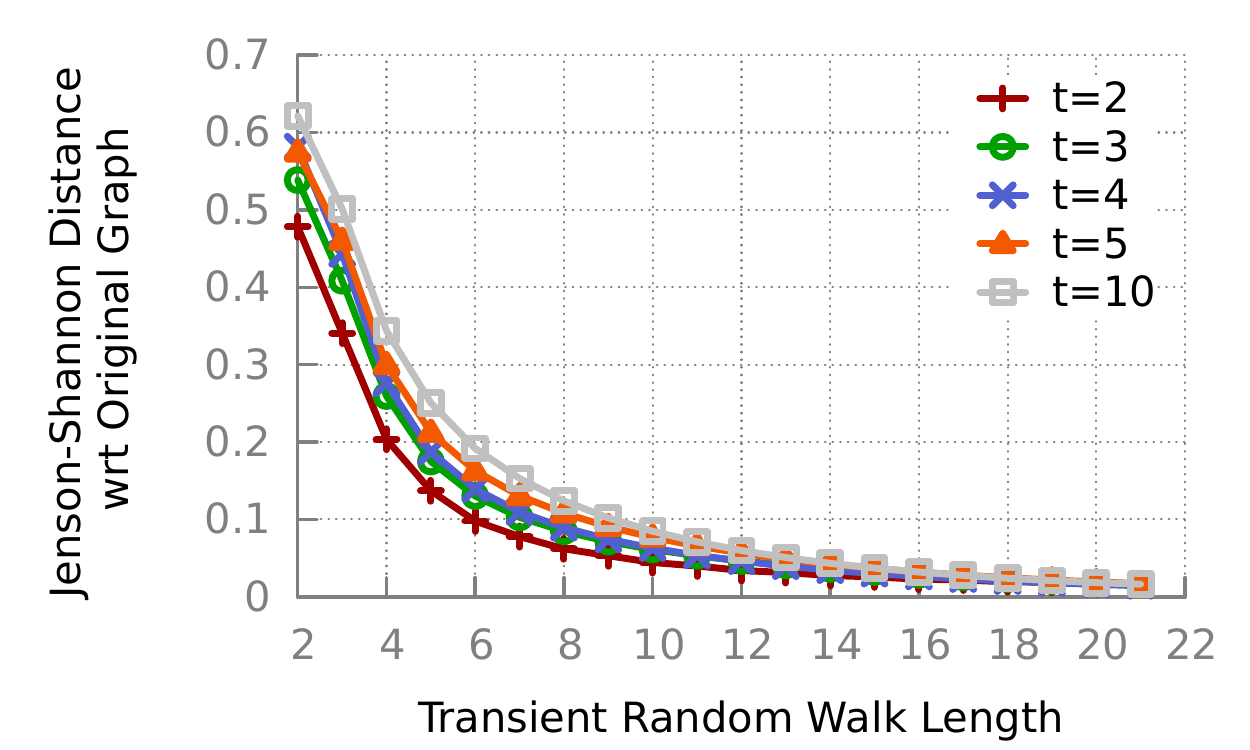,width=0.33 \textwidth}\vspace{-0.00in}\\
{(b)}
\end{tabular}
}
\caption{{\em Jenson-Shannon distance between transient probability distributions for original graph and transformed graph} using 
(a) Facebook interaction graph (b) Facebook wall post graph. We can see that as the original graph is perturbed to a larger degree, 
the distance between original and transformed transient distributions increases, decreasing application utility.}
\label{fig:utility-jenson-shannon}
\vspace{-0.1in}
\end{figure*}

\begin{figure*}[htp]
\centering
\mbox{
\hspace{-0.2in}
\hspace{-0.12in}
\begin{tabular}{c}
\psfig{figure=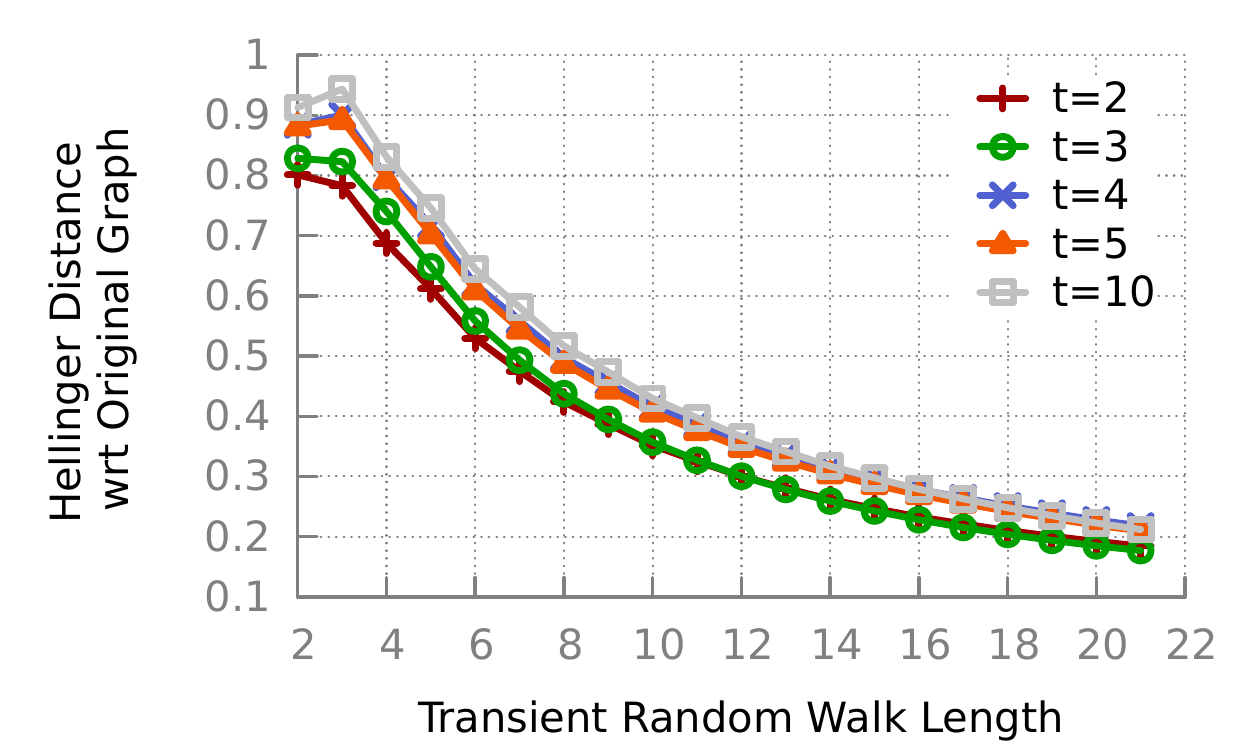,width=0.33 \textwidth}\vspace{-0.00in}\\
{(a)}
\end{tabular}
\begin{tabular}{c}
\psfig{figure=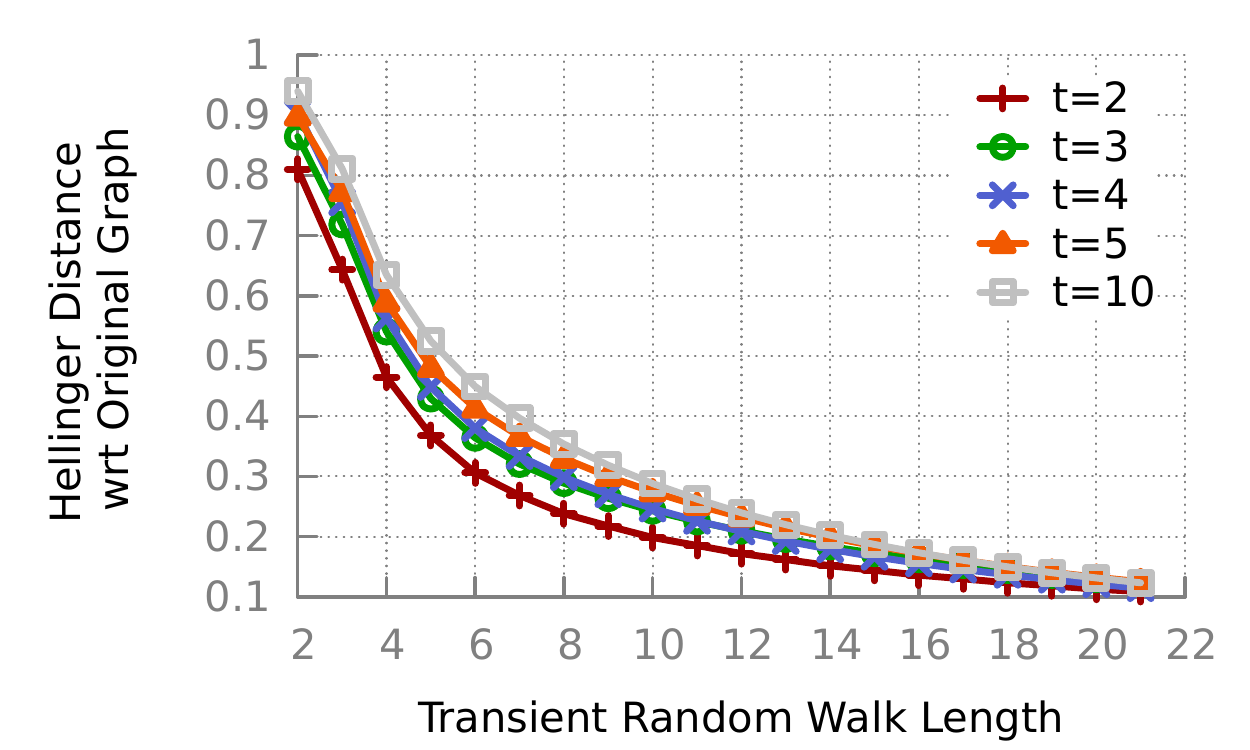,width=0.33 \textwidth}\vspace{-0.00in}\\
{(b)}
\end{tabular}
}
\caption{{\em Hellinger distance between transient probability distributions for original graph and transformed graph} using 
(a) Facebook interaction graph (b) Facebook wall post graph. We can see that even with different notions of distance between probability 
distributions (Hellinger/Jenson-Shannon), the distance between original graph and perturbed graph monotonically increases depending 
on the perturbation degree. }
\label{fig:utility-hellinger}
\vspace{-0.1in}
\end{figure*}

\subsection{Algorithm Analysis}

Our above results show the general relationship between our utility metrics and global graph properties (which hold for any 
perturbation algorithm). Next, we analyze the properties of our proposed perturbation algorithm. 

First, we empirically compute the mean vertex utility of the perturbed graphs ($VU_{mean}$), for varying perturbation 
parameters and varying application parameters. Figure~\ref{fig:utility-jenson-shannon} depicts the mean vertex utility
for the Facebook interaction and friendship graphs using the Jenson-Shannon information theoretic distance metric . We can 
see that as the perturbation parameter increases, the distance metric increases. This is not surprising, since additional 
noise will increase the distance between probability distributions computed from original and perturbed graphs. We can also 
see that as the application parameter $l$ increases, the distance metric decreases. This illustrates that our perturbation 
algorithm is ideally suited for security applications that rely on local or global community structures, as opposed to 
applications that require \emph{exact} information about one or two hop neighborhoods. We can see a similar trend when using 
Hellinger distance to compute the distance between probability distributions, as shown in Figure~\ref{fig:utility-hellinger}. 

\begin{mythm}
The expected degree of each node after the perturbation algorithm is the same as in the original graph: $\forall v \in V, 
E(deg(v)') = deg(v)$, where $deg(v)'$ denotes the degree of vertex $v$ in $G'$.
\end{mythm}

\begin{proof}
On an expectation, half the degree of any node $v$ is preserved via outgoing random walks from $v$ in the perturbation 
process. To prove the theorem, we need to show that for each node is the terminal point of deg(v)/2 random walks in 
the perturbation mechanisms (on average). From the time reversibility property of the random walks, we have that 
$P_{ij}^t \pi_i = P_{ji}^t \pi_j$. Thus for any node $i$, the incoming probability of a random walk starting from node 
j is $P_{ji}^t = P_{ij}^t \frac{deg(i)}{deg(j)}$, i.e., it is proportional to the node degree of $i$. Thus the expected 
number of random walks terminating at node $i$ in the perturbation algorithm is given by $\sum_{v \in V} deg(v) P_{vi}^{t-1}$.
This is equivalent to $\sum_{v \in V}  P{iv}^{t-1} deg(i) = deg(i)$. Since half of these walks will be added to the graph 
G' on average, we have that $E(deg(v)') = deg(v)$. 
\end{proof}

\begin{mycor}
The expected value of the largest eigenvalue of the transformed graph is bounded as $\max(d_{avg},\sqrt(d_{max})) \leq E(\lambda_1') \leq d_{max}$
\end{mycor}

From the Perron-Frobenius theorem, we have that the largest eigenvalue of the graph is related to the notion of 
average graph degree as follows:

\begin{align}
\max(d_{avg}',\sqrt(d_{max}')) \leq \lambda_1' \leq d_{max}'
\end{align}

Taking expectation on the above equation, and using the previous theorem yields the corollary. 

Next, we show experimental results validating our theorem. Figure~\ref{fig:degree-dist} depicts the node degrees 
of the original graphs, and expected node degrees of the perturbed graphs, corresponding to all nodes in the Facebook 
interaction and friendship graphs. In this figure, the points in a vertical line for different perturbed graphs 
correspond to the same node index. We can see that the degree distributions are nearly identical, validating our 
theoretical results.  

\begin{figure*}[htp]
\centering
\mbox{
\hspace{-0.2in}
\hspace{-0.12in}
\begin{tabular}{c}
\psfig{figure=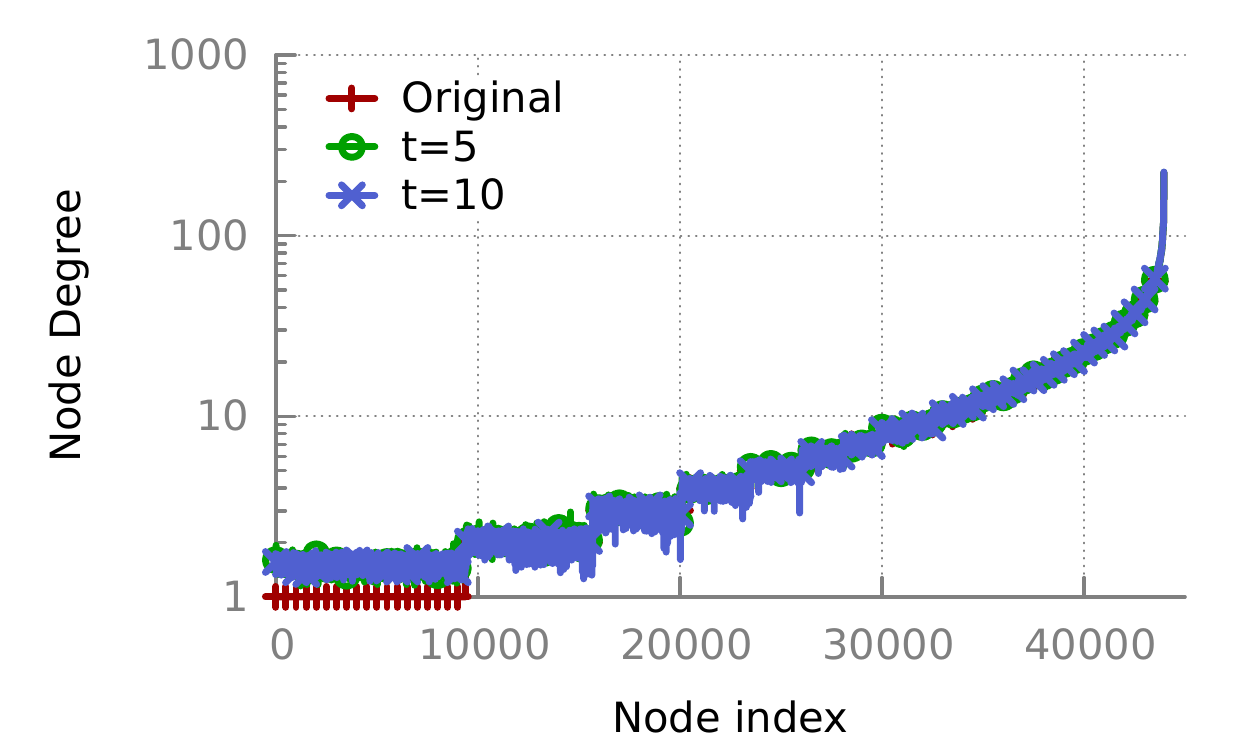,width=0.33 \textwidth}\vspace{-0.00in}\\
{(a)}
\end{tabular}
\begin{tabular}{c}
\psfig{figure=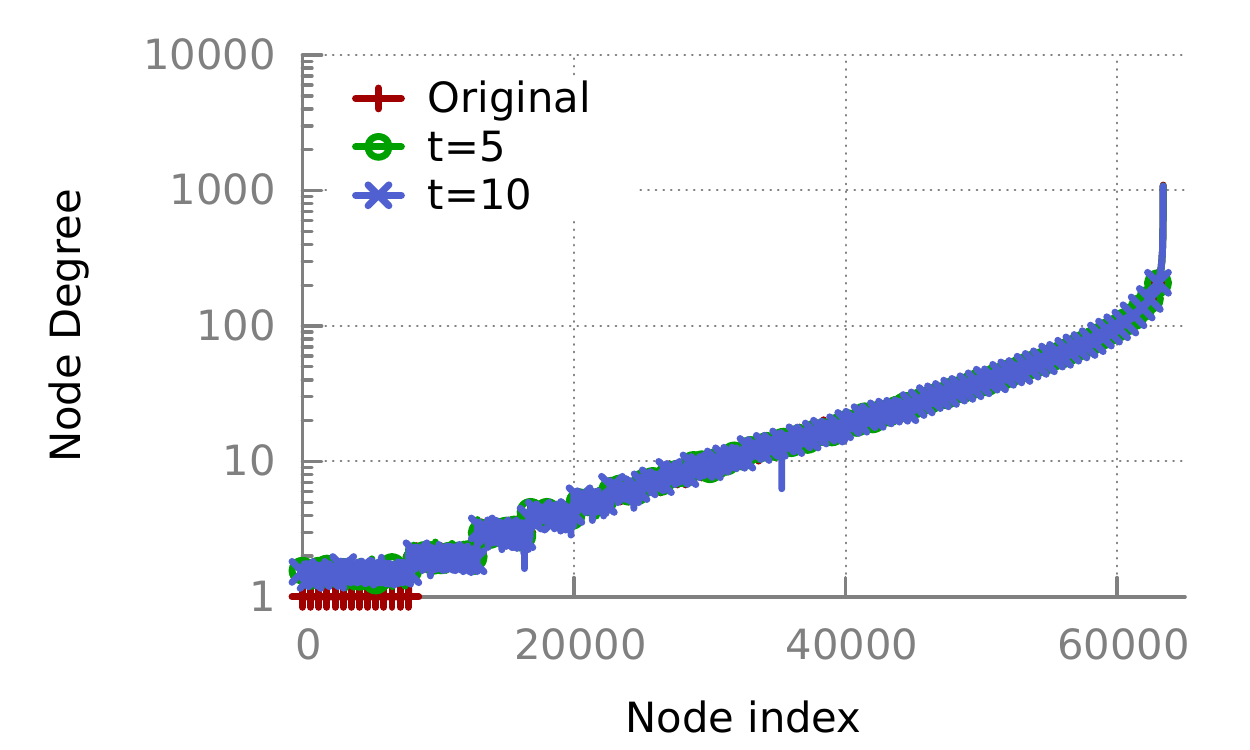,width=0.33 \textwidth}\vspace{-0.00in}\\
{(b)}
\end{tabular}
}
\caption{{\em Degree distribution of nodes} using (a) Facebook interaction graph (b) Facebook wall post graph. We can see that the expected degree 
of each node after the perturbation process remains the same as in the original graph.}
\label{fig:degree-dist}
\vspace{-0.1in}
\end{figure*}

\begin{figure*}[htp]
\centering
\mbox{
\hspace{-0.2in}
\hspace{-0.12in}
\begin{tabular}{c}
\psfig{figure=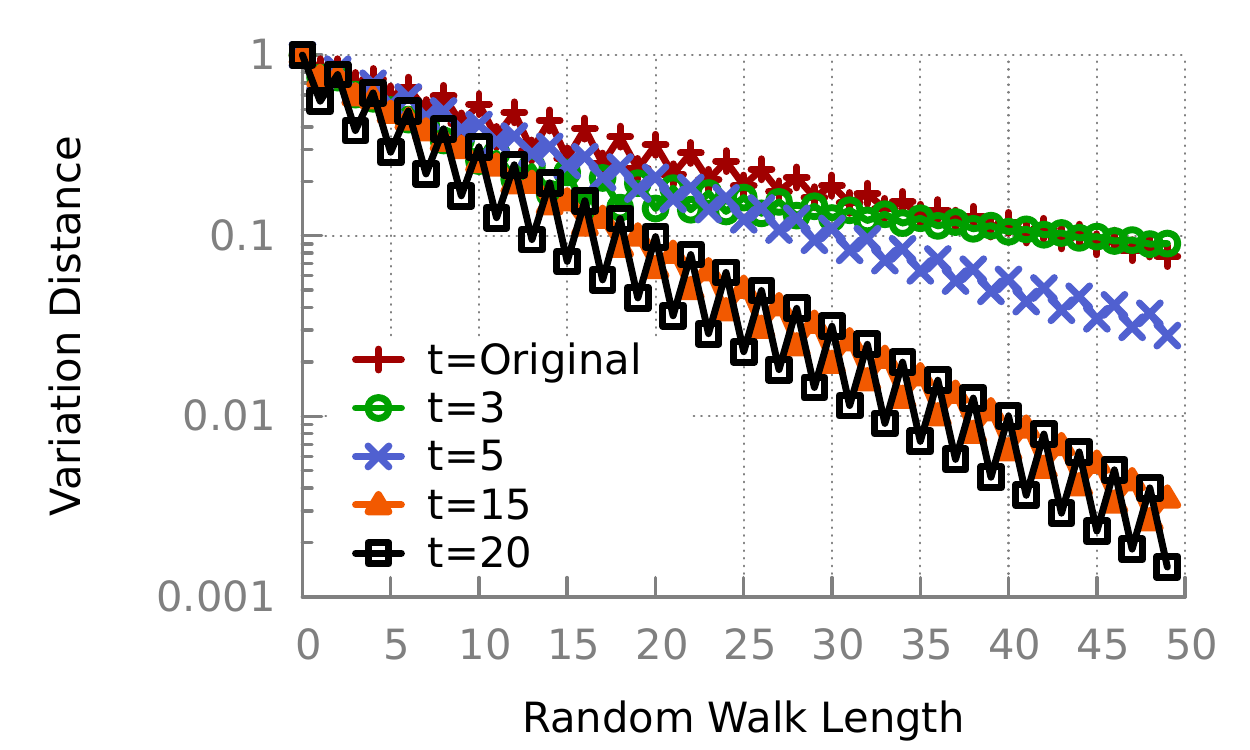,width=0.33 \textwidth}\vspace{-0.00in}\\
{(a)}
\end{tabular}
\begin{tabular}{c}
\psfig{figure=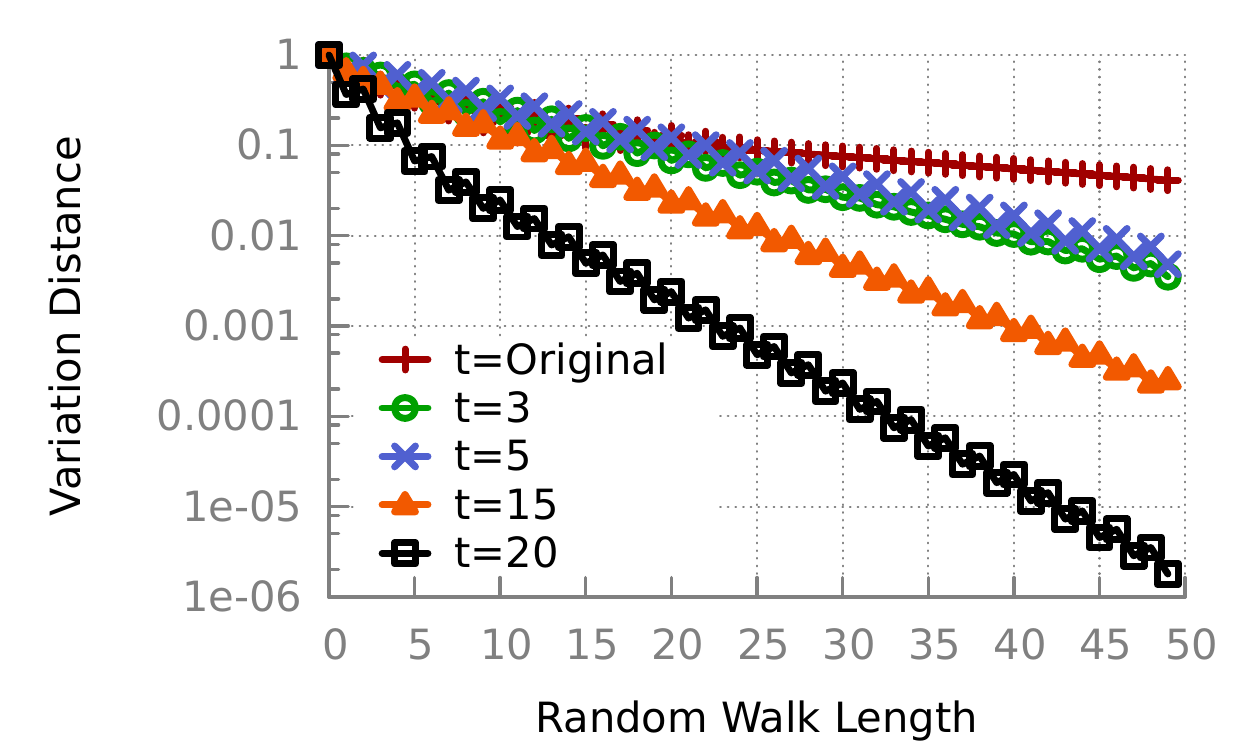,width=0.33 \textwidth}\vspace{-0.00in}\\
{(b)}
\end{tabular}
}
\caption{{\em Total variation distance as a function of random walk length} using (a) Facebook interaction graph (b) Facebook wall post graph. We can see that increasing the 
perturbation parameter of our algorithm reduces the mixing time of the graph.}
\label{fig:mixing}
\vspace{-0.1in}
\end{figure*}

\begin{mythm}
Using our perturbation algorithm, the mixing time of the perturbed graphs is 
related to the mixing time of the original graph as 
follows: $\frac{\tau_G(\epsilon)}{t} \leq E(\tau_{G'}(\epsilon)) \leq \tau_G(\epsilon)$.
\label{thm:mixing}
\end{mythm}

Theorem~\ref{thm:mixing} bounds the mixing time of the perturbed graph using the mixing time of the original graph and 
the perturbation parameter $t$. We defer the proof to the Appendix. 
Finally, we compute the mixing time of the original and perturbed graphs using simulations. Figure~\ref{fig:mixing} depicts 
the total variation distance between random walks of length $x$ and the stationary distribution, for the original and 
perturbed graphs.~\footnote{Variation distance has a slight oscillating behavior at odd and even steps of the random walk; this phenomenon 
is also observed in Figure~\ref{fig:privacy-median}.} We can see that as the perturbation parameter increases, the 
total variation distance (and the mixing time) decreases. Moreover, for small values of the perturbation parameter, the difference from 
the original topology is small. As an aside, it is interesting to note that the variation distance for the Facebook 
friendship graph is orders of magnitude smaller that the Facebook interaction graph. This is because the Facebook interaction 
graphs are a lot sparser, resulting in slow mixing. 

%% file: privacy.tex
\section{Privacy}
\label{sec:privacy}

We now address the question of understanding link privacy of our 
perturbation algorithm. We use several notions for quantifying 
link privacy, which fall into two categories (a) quantifying 
exact probabilities of de-anonymizing a link given specific 
adversarial priors, and (b) quantifying risk of de-anonymizing a link 
without making specific assumptions about adversarial priors. 
We also characterize the relationship between utility and privacy 
of a perturbed graph. 

\subsection{Bayesian formulation for link privacy}

\begin{mydef}
We define the privacy of a link L (or a subgraph) in the original graph, as the 
probability of existence of the link (or a subgraph), as computed by the adversary, 
under an optimal attack strategy using its prior information H: $P(L|G',H)$
\end{mydef}

Note that low values of link probability $P(L|G',H)$ correspond to high privacy. We 
cast the problem of computing the link probability as a Bayesian inference problem. 
Using Bayes theorem, we have that:

\begin{align}
P (L | G', H) = \frac{P(G'|L,H) \cdot P(L|H)}{P(G'|H)}
\end{align}

In the above expression, $P(L|H)$ is the prior probability of the link. In Bayesian 
inference, $P(G'|H)$ is a normalization constant that is typically difficult to compute, 
but this is not an impediment for the analysis since sampling techniques can be used 
(as long as the numerator of the Bayesian formulation is computable upto a constant factor)~\cite{metropolis,MetropolisHastings}.  
Our key theoretical challenge is to compute $P(G'|L,H)$.

To compute $P(G'|L,H)$, the adversary has to consider all possible graphs $Gp$, 
which have the link $L$, and are consistent with background information $H$. 
Thus, we have that: 

\begin{align}
\label{eqn:gp}
P(G'|L,H) = \sum_{Gp} P(G'| Gp) \cdot P(Gp | L,H)
\end{align}

The adversary can compute $P(G' | Gp)$ using the knowledge of the perturbation 
algorithm; we assume that the adversary knows the full details of our perturbation 
algorithm, including the perturbation parameter $t$. Observe that given $Gp$, 
edges in $G'$ can be modeled as samples from the probability distribution of 
$t$ hop random walks from vertices in $Gp$. Thus we can compute $P(G'|Gp)$ as follows:

We can compute $P(G' | Gp)$ 
using the $t$ hop transition probabilities of vertices in $Gp$. 

\begin{align}
P(G' | Gp) = {(\frac{1}{2})}^{2 \cdot m} \cdot \binom{2 \cdot m}{m'} \cdot \nonumber \\ \prod_{i-j \in E(G')} \frac{P_{ij}^t(Gp) + P_{ji}^t(Gp)}{2}
\end{align}

In general, the number of possible graphs $Gp$ that have L as a link and are 
consistent with the adversary's background information can be very large, 
and the computation in Equation~\ref{eqn:gp} then becomes intractable. %

\begin{figure}[!tp]
\centering
\includegraphics[width=0.4\textwidth]{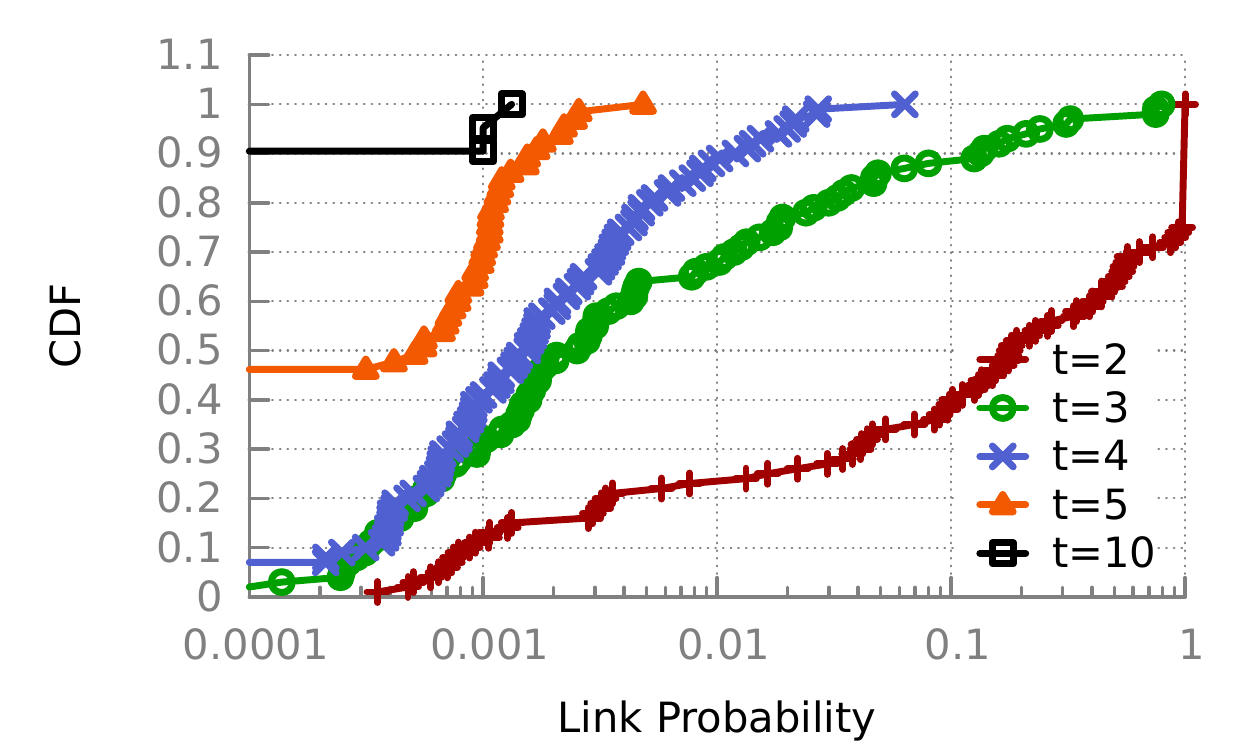}
\vspace{-0.1in}
\caption{{\em Cumulative distribution of link probability $P(L | G',H)$} (x-axis is logscale) 
under worst case prior $H = G-L$ using a synthetic scale free topology. Small probabilities 
offer higher privacy protection. }
\label{fig:privacy-worstprior}
\vspace{-0.1in}
\end{figure}

For evaluation, we consider a special case of this definition: the adversary's prior 
is the entire original graph without the link $L$ (which is the link for which we want 
to quantify privacy). Observe that this is a very powerful adversarial prior; we use 
this prior to shed light on the \emph{worst-case} link privacy using our perturbation algorithm. 
Under this prior, we have that:

\begin{align}
P (L | G', G - L)  & =   \frac{P(G'|L,G - L) \cdot P(L|G - L)}{P(G'|G - L)} \nonumber \\
& = \frac{P(G' | G) \cdot P(L | G - L)}{P(G' | G - L)} \nonumber \\
& = \frac{P(G' | G) \cdot P( L | G - L)} { \sum_l P(G' | G - L + l) }
\end{align}

Using $G-L$ as the adversarial prior constraints the set of possible $Gp$ to a polynomial number. However, 
even in this setting, we found that the above computation is computationally expensive ($> O(n^3)$) 
using our real world social networks. Thus to get an understanding of link privacy in this setting, 
we generated a 500 node synthetic scale-free topology using the preferential attachment methodology 
of Nagaraja~\cite{nagaraja:pet07}. 
The parameters of the scale free topology was set using the average degree in the Facebook 
interaction graph. Figure~\ref{fig:privacy-worstprior} depicts the cumulative distribution for 
link probability (probability of de-anonymizing a link, $P(L|G',H)$) in this setting 
(worst case prior) for the synthetic scale-free topology. We can see that there is significant 
variance in the privacy protection 
received by links in the topology: for example, using perturbation parameter $t=2$, $40\%$ 
of the links have a link probability less than $0.1$ (small is better), while $30\%$ of the 
links have have a probability of 1 (can be fully de-anonymized). Note that this is a worst case 
analysis, since we assume that an attacker knows the entire original graph except the link in 
question. Furthermore, even in this setting, as the perturbation parameter $t$ increases, the privacy 
protection received by links substantially improves: for example, using $t=5$, all of the links have a link  
probability less than $0.01$, while $70\%$ of the links have a link probability of less 
than $0.001$. Thus we can see that even in this worst case analysis, our perturbation mechanism
offers good privacy protection.  

Comparison with Hay et al~\cite{hay:umass07}: previous work proposed a perturbation approach 
where $k$ real edges are deleted and $k$ fake edges are introduced at random. Even considering 
$k=m/2$ (which would substantially hurt utility, for example, by introducing large number of edges 
between Sybils and honest users), 50\% of the edges in the perturbed graph 
are \emph{real} edges between users; for these edges, $P(L|G')=0.5$. Here we can see the benefit 
of our perturbation mechanism: by sampling fake edges based on the \emph{structure} of the original graph, 
we are able to significantly improve link privacy without hurting utility. 

\subsection{Relationship between privacy and utility}

Intuitively, there is a relationship between link privacy and the utility of the 
perturbed graphs. Next, we formally quantify this relationship.  %

\begin{mythm}
Let the maximum vertex utility of the graph (over all vertices) corresponding to an application parameter 
$l$ be $VU_{max}(G,G',l)$. Then for any two pair of vertices $A$ and $B$, we have that 
$P(L_{AB} | G') \geq f(\delta)$, where $f(\delta)$ denotes the prior probability of two vertices being 
friends given that they are both contained in a $\delta$ hop neighborhood, and $\delta$ is computed 
as $\delta = \min {k: P_{AB}^k(G') - VU_{max}(G,G',k) > 0}$.
\label{thm:utility-privacy}
\end{mythm}

Utility measures the change in graph structure between original 
and perturbed graphs. If this change is small (high utility), then an adversary 
can infer information about the original graph given the perturbed graph. 
For a given level of utility, the above theorem demonstrates a lower bound on 
link privacy. 

\begin{figure*}[htp]
\centering
\mbox{
\hspace{-0.2in}
\hspace{-0.12in}
\begin{tabular}{c}
\psfig{figure=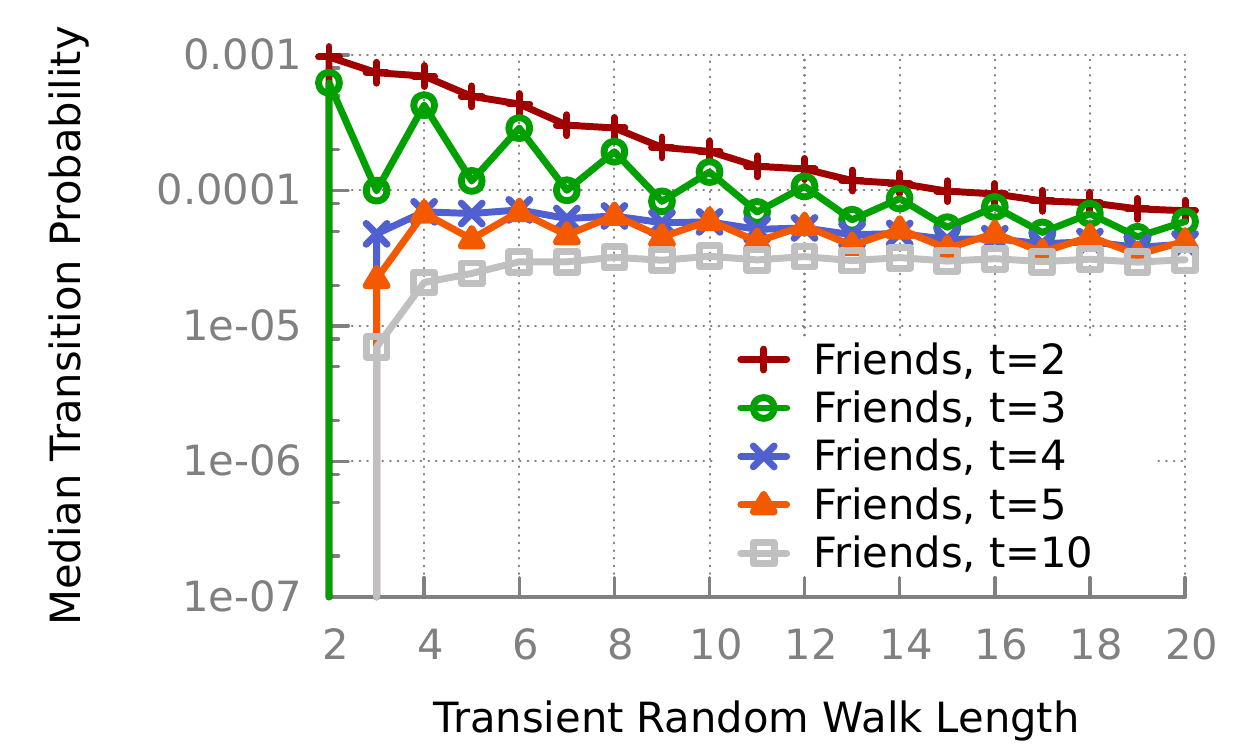,width=0.33 \textwidth}\vspace{-0.00in}\\
{(a)}
\end{tabular}
\begin{tabular}{c}
\psfig{figure=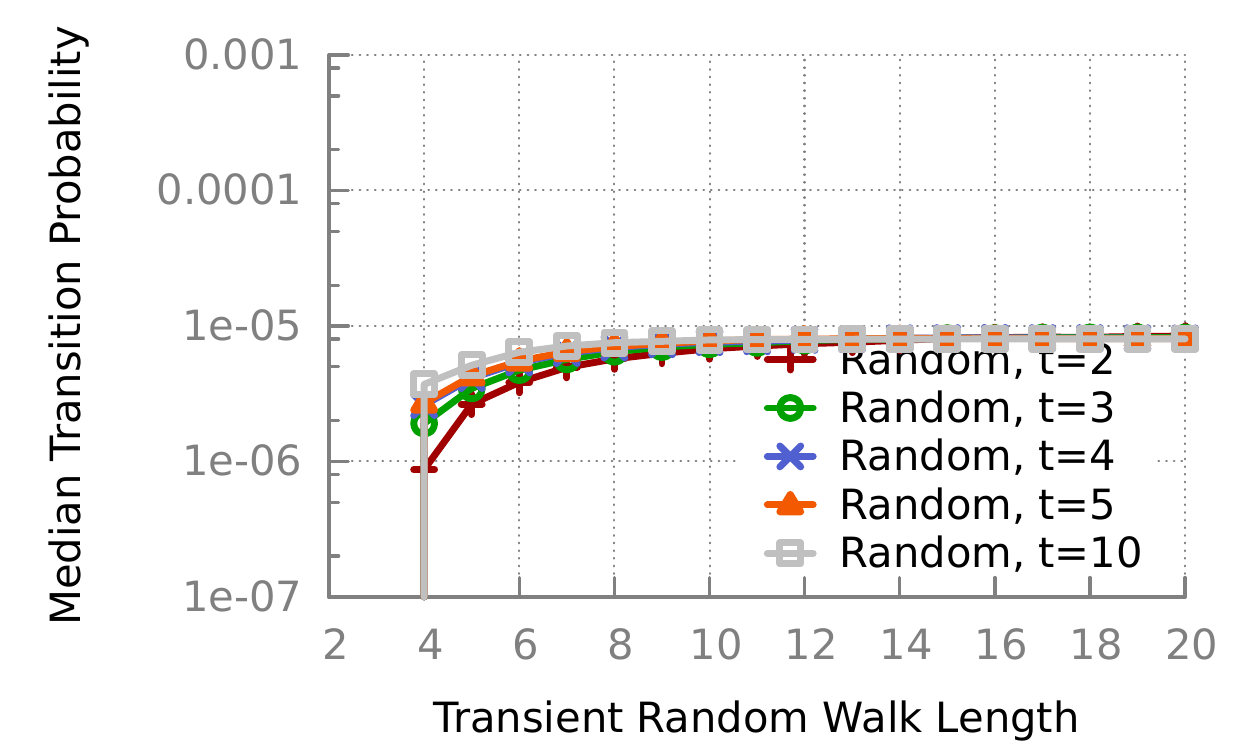,width=0.33 \textwidth}\vspace{-0.00in}\\
{(b)}
\end{tabular}
}
\vspace{-0.1in}
\caption{{\em Median transition probability between two vertices in transformed graph} when (a) two vertices were neighbors (friends) in the original graph and (b) two vertices were not neighbors in the original graph, for the 
Facebook interaction graph. }
\label{fig:privacy-median}
\vspace{-0.1in}
\end{figure*}

\begin{figure*}[htp]
\centering
\mbox{
\hspace{-0.2in}
\hspace{-0.12in}
\begin{tabular}{c}
\psfig{figure=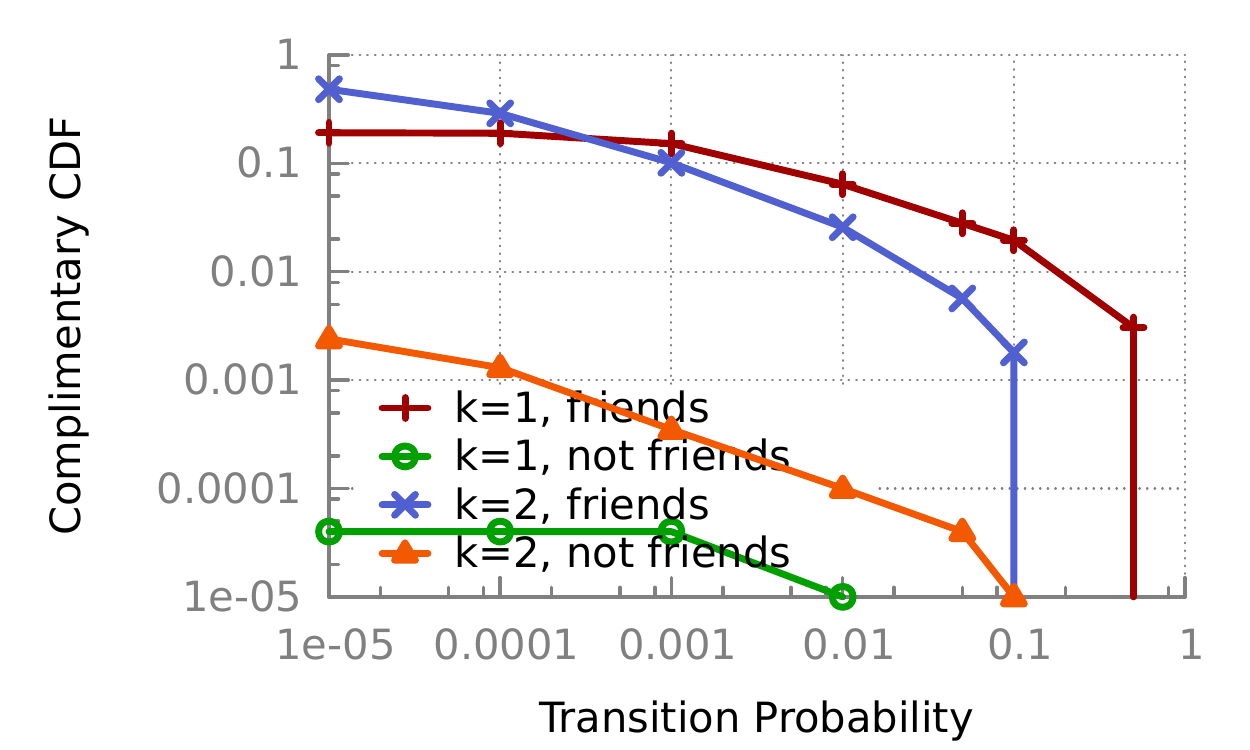,width=0.3 \textwidth}\vspace{-0.00in}\\
{(a)}
\end{tabular}
\begin{tabular}{c}
\psfig{figure=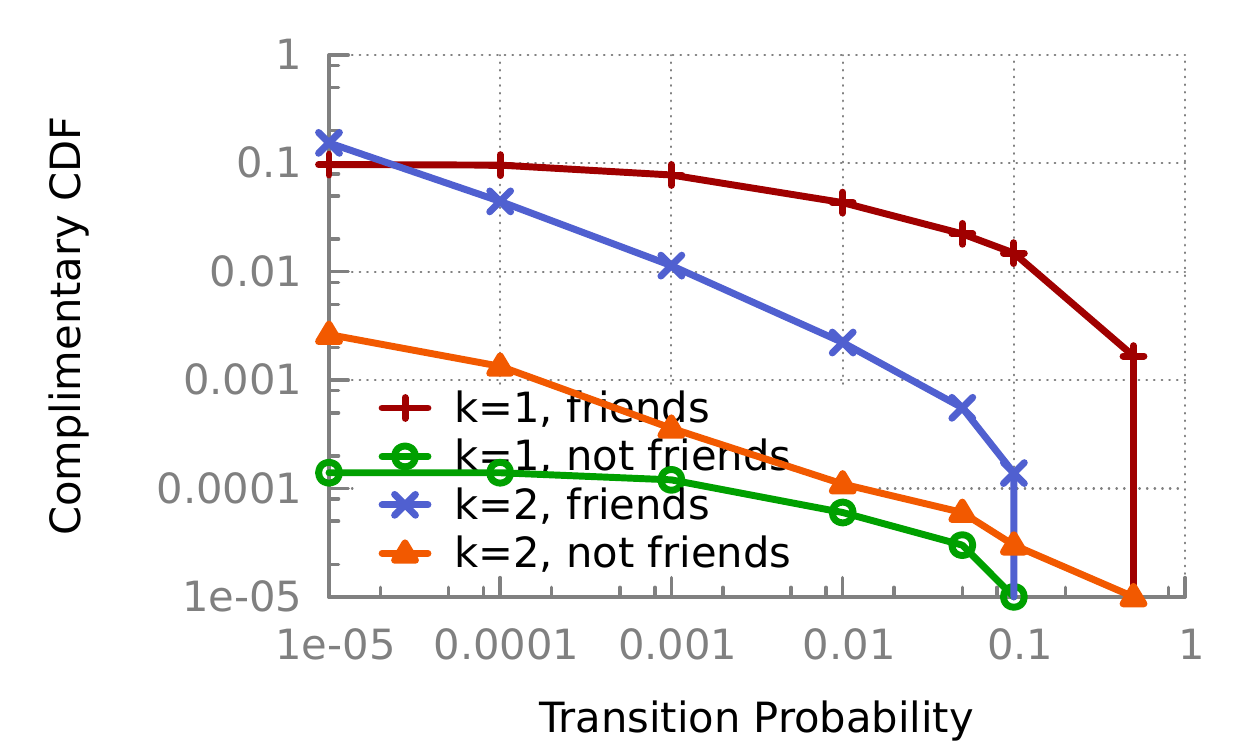,width=0.3 \textwidth}\vspace{-0.00in}\\
{(b)}
\end{tabular}
\begin{tabular}{c}
\psfig{figure=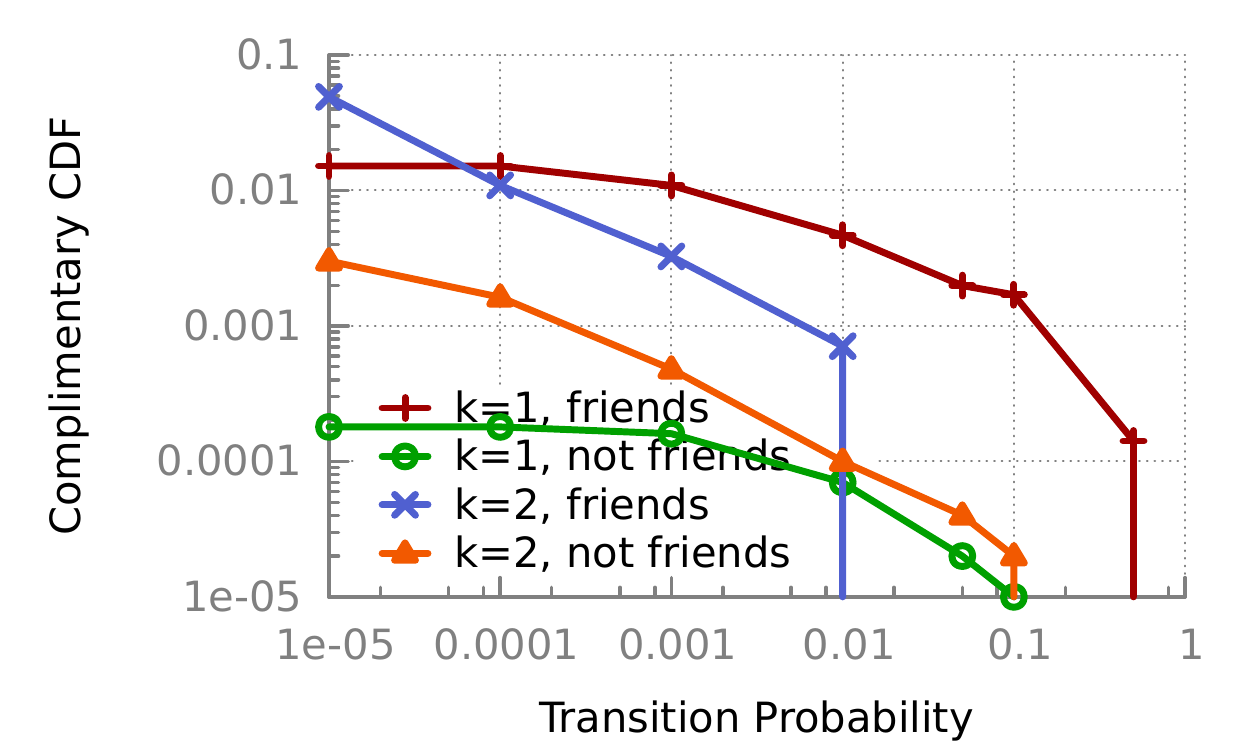,width=0.3 \textwidth}\vspace{-0.00in}\\
{(b)}
\end{tabular}
}
\vspace{-0.1in}
\caption{{\em Complimentary cumulative distribution for $P_{AB}^k(G')$} using (a) perturbation parameter $t=2$, and (b) perturbation parameter $t=5$, and (c) perturbation 
parameter $t=10$ for the Facebook interaction graph. }
\label{fig:privacy-fulldist}
\vspace{-0.1in}
\end{figure*}

We defer the proof of the above theorem to the Appendix. The above theorem is a general theorem 
that holds for all perturbed graphs. To shed some intuition, we specifically analyze the lower 
bounds on privacy for our perturbation algorithm (where parameter $t$ governs utility), 
using the transition probability between $A$ and $B$ in the perturbed graph $G'$ as a feature 
to assign probabilities to nodes $A$ and $B$ of being friends in the original graph $G$. 
We are interested in the quantity $P(L_{AB} \mid {P}_{AB}^k(G') > x)$, 
for different values of $k$. Using Bayes' theorem, we have that: 

\begin{align}
P( L_{AB} \mid P_{AB}^k(G') > x) & = \frac{P(P_{AB}^k(G') > x \mid L_{AB}) \cdot P(L_{AB})}{P(P_{AB}^k(G') > x)} 
\label{eqn:sim-privacy}
\end{align}

Also, we have that: 

\begin{align}
P(P_{AB}^k(G') > x) = P(P_{AB}^k(G') > x \mid L_{AB}) \cdot P(L_{AB}) + \nonumber \\ P(P_{AB}^k(G') > x \mid \overline{L}_{AB}) \cdot P(\overline{L}_{AB}) 
\end{align}

where $\overline{L}_{AB}$ denotes the event when vertices $A$ and $B$ don't have a link. 

To get some insight, we computed the probability distributions $P(P_{AB}^k(G') | L_{AB})$ and $P(P_{AB}^k(G') | \overline{L}_{AB})$ using 
simulations on our real world social network topologies. Figure~\ref{fig:privacy-median} depicts the median value of the respective 
probability distributions, as a function of parameter $k$, for different perturbation parameters $t$, using the Facebook interaction graph. 
We can see that the median transition probabilities are higher for the scenario where two users are originally friends (as opposed to 
the setting where they are not friends). We can also see that the difference between median transition values in the two 
scenarios is higher when (a) the perturbation parameter $t$ is small, and (b) the parameter $K$ (random walk length) is small. 
This difference is related to the privacy of the link - larger the difference, greater the loss in privacy. The insight from this 
figure is that for small perturbation parameters, the closer two nodes are to each other in the perturbed graph, the higher their 
chances of being friends in the original graph. %
Next, we consider the full distribution of transition probabilities (as opposed to only the median values discussed above). 
Figure~\ref{fig:privacy-fulldist}(a) depicts the complimentary cumulative distribution ($P_{AB}^k(G') > x)$ using perturbation parameter $t=2$ 
for the Facebook interaction graph. Again, we can see that when $A$ and $B$ are friends, they have higher transition probability to each other 
in the perturbed graph, compared to the scenario when $A$ and $B$ are not friends. Moreover, as the value of $k$ increases, the gap between the 
distributions becomes smaller. Similarly, as the value of $t$ increases, the gap between the distributions becomes smaller, as depicted 
in Figures~\ref{fig:privacy-fulldist}(b-c). We can use these simulation results to compute the probabilities in Equation~\ref{eqn:sim-privacy}. 
One way to analyzing the lower bound on link privacy would be to choose a uniform prior for vertices being friends in the 
original graph, i.e., $P(L_{AB})=\frac{m}{\binom{n}{2}}$. Correspondingly, $P(\overline{L}_{AB}) = 1-P(L_{AB})$.

\begin{figure}[!tp]
\centering
\includegraphics[width=0.4\textwidth]{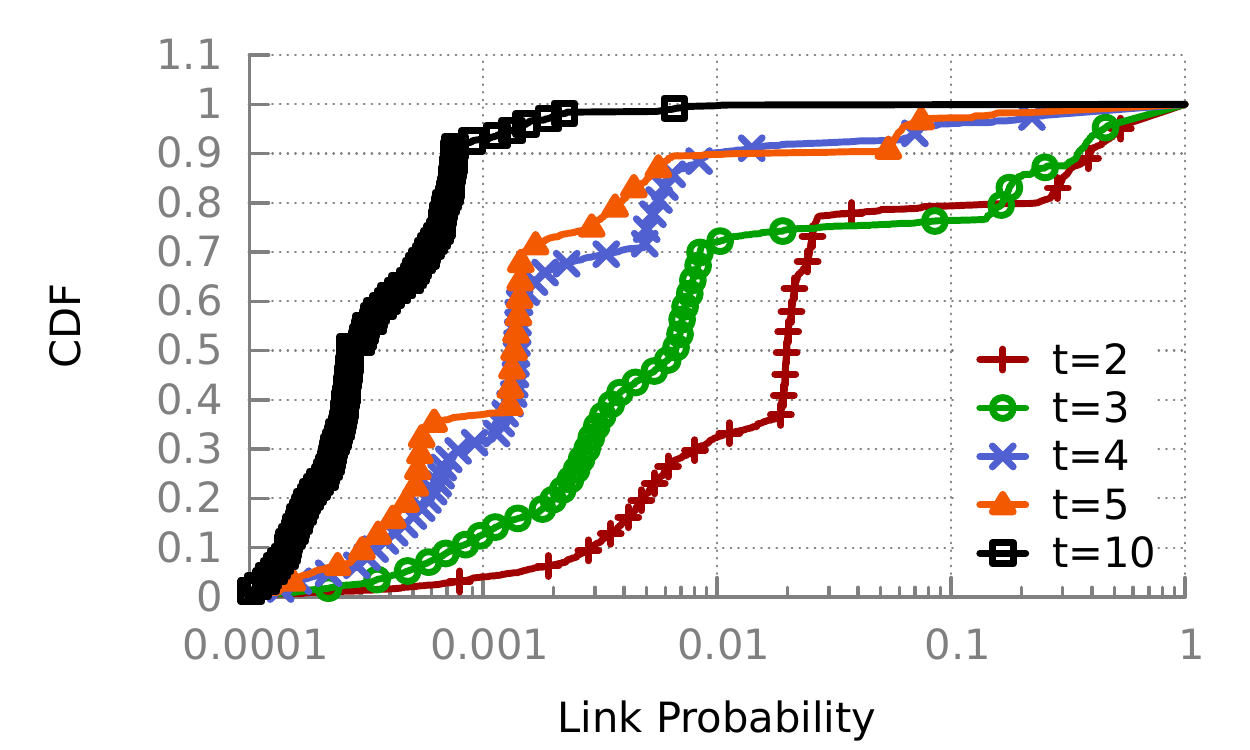}
\vspace{-0.1in}
\caption{{\em Cumulative distribution of link probability $P(L| G',H)$} for the Facebook interaction graph. Smaller probabilities offer higher 
privacy protection}
\label{fig:privacy-lowerbound}
\vspace{-0.1in}
\end{figure}

Figure~\ref{fig:privacy-lowerbound} depicts the cumulative distribution of link probability computed using the above methodology, for the 
Facebook interaction graph. We can see the variance in privacy protection received by links in the topology. Using $t=2$, $80\%$ of the links 
have a link probability of less than $0.1$. Increasing perturbation parameter $t$ significantly improves performance; using $t=3$, 
$95\%$ of the links have a link probability less than $0.1$, and $90\%$ of the links have a link probability less than $0.01$. 
In summary, our analysis shows that a given level of utility translates into a lowerbound on privacy offered by the perturbation mechanism.

\subsection{Risk based formulation for link privacy}

\begin{figure*}[!tp]
\centering
\mbox{
\hspace{-0.2in}
\hspace{-0.12in}
\begin{tabular}{c}
\psfig{figure=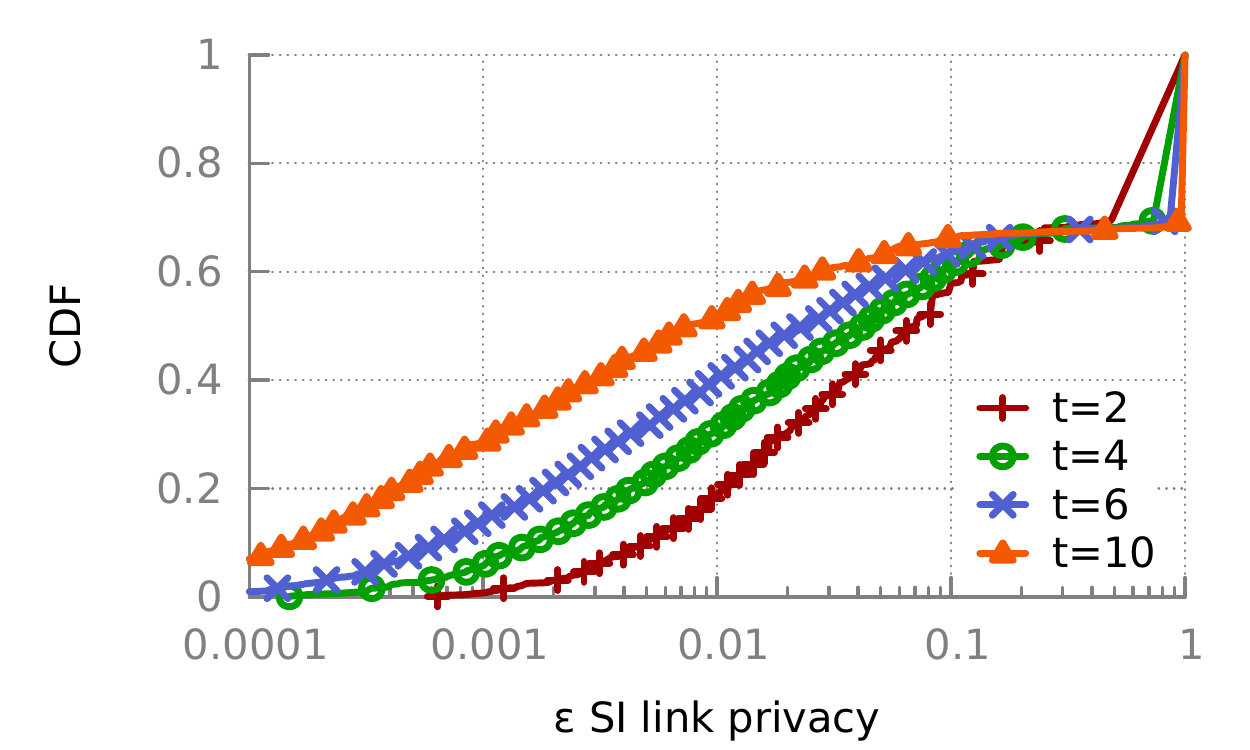,width=0.33 \textwidth}\vspace{-0.00in}\\
{(a)}
\end{tabular}
\begin{tabular}{c}
\psfig{figure=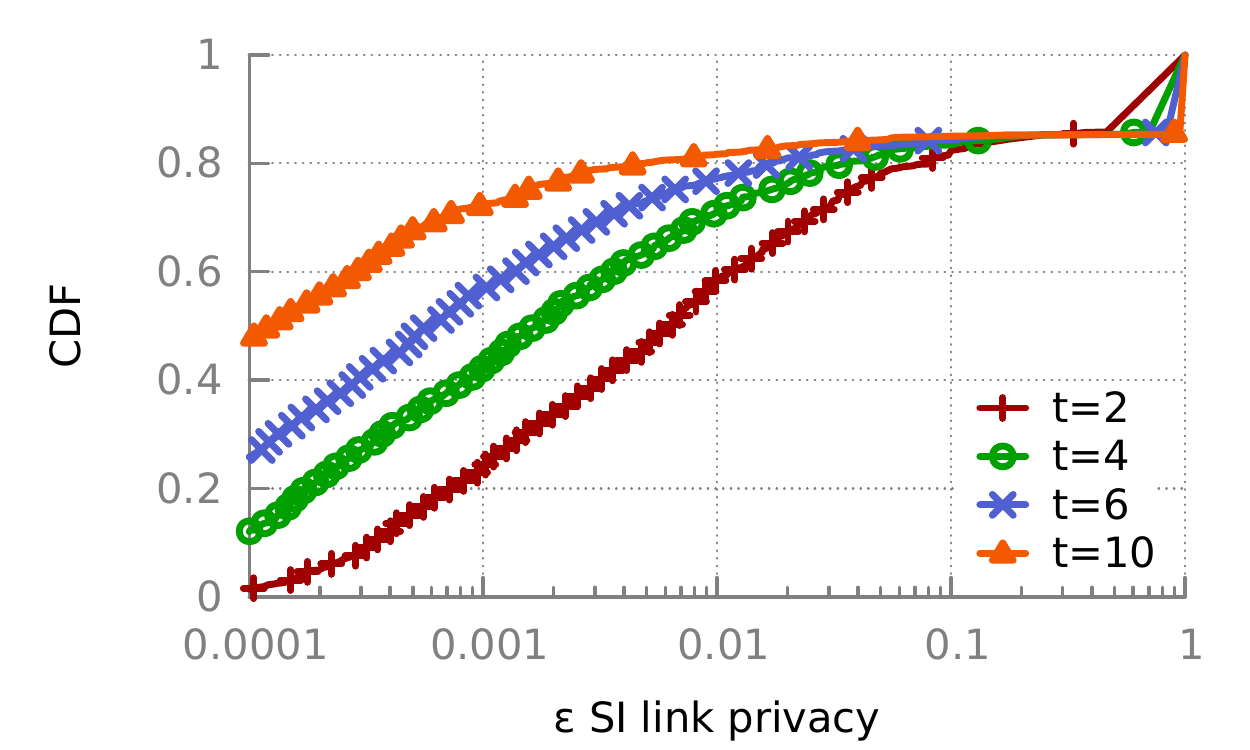,width=0.33 \textwidth}\vspace{-0.00in}\\
{(b)}
\end{tabular}
}
\vspace{-0.1in}
\caption{{\em Cumulative distribution of $\epsilon$ SI link privacy} for (a) Facebook interaction graph and (b) Facebook link graph. Note that the SI privacy definition is not applicable to links for which 
either of the vertex has degree of 1, since removing that link disconnects the graph.}
\label{fig:si-privacy}
\vspace{-0.1in}
\end{figure*}

\begin{figure*}[!tp]
\centering
\mbox{
\hspace{-0.2in}
\hspace{-0.12in}
\begin{tabular}{c}
\psfig{figure=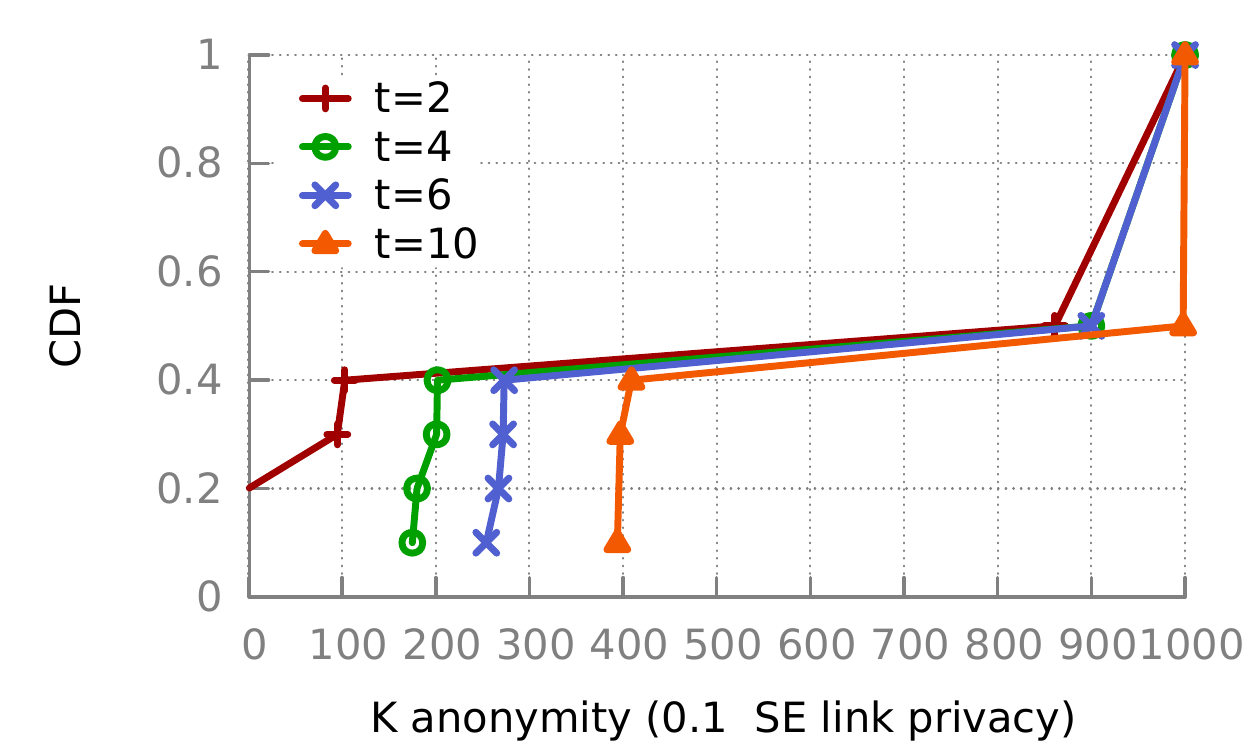,width=0.33 \textwidth}\vspace{-0.00in}\\
{(a)}
\end{tabular}
\begin{tabular}{c}
\psfig{figure=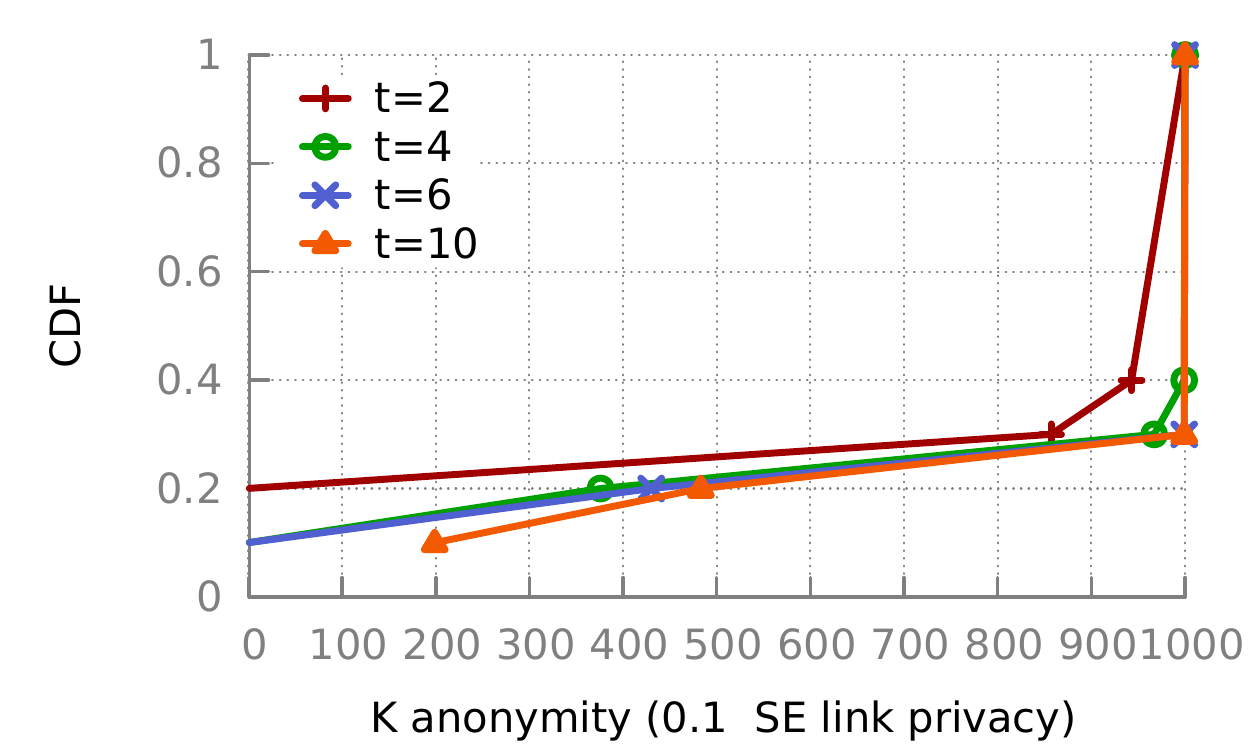,width=0.33 \textwidth}\vspace{-0.00in}\\
{(b)}
\end{tabular}
}
\vspace{-0.1in}
\caption{{\em Cumulative distribution of K-anonymity set size for $\epsilon=0.1$ SE link privacy} using  (a) Facebook interaction graph and (b) Facebook friendship graph.}
\label{fig:se-privacy}
\vspace{-0.1in}
\end{figure*}

The dependence of the Bayesian inference based privacy definitions on the prior of the adversary 
motivates the formulation of new metrics that are not specific to adversarial priors. We first 
illustrate a preliminary definition of link privacy (which is unable to account for links to degree 1 
vertices), and then subsequently improve it. 

\begin{mydef}
 We define the structural %
impact (SI) of a link L in graph G with respect to a 
perturbation mechanism $M$, as the statistical distance between probability distributions 
of the output of the perturbation mechanism (i.e., the set of possible perturbed graphs)
when using (a) the original graph G as an input to the perturbation mechanism, and (b) 
the graph G - L as an input to the perturbation mechanism.  Let $P(G'=M(G))$ denote the 
probability \emph{distribution} of perturbed graphs $G'$ using the perturbation mechanism M and 
input graph G. A link has $\epsilon$ SI-privacy if the statistical distance 
$|| P(G'=M(G)) - P(G'=M(G-L))|| < \epsilon$.
\end{mydef}

Intuitively, if the SI of a link is high, then the perturbation process leaks more 
information about that link. On the other hand, if the SI of a link is low, then 
the perturbed graph G' leaks less information about the link. 

As before, we consider the total variation distance as our distance metric between probability 
distributions. Observe that the links in graphs G' are samples from the probability distribution 
$P_v^t(G)$, for $v \in V$. So we can bound the difference in probability distributions of perturbed 
graphs generated from G and G-L, by the worst case total variation distance between probability 
distributions $P_v^t(G)$ and $P_v^t(G-L)$, over all $v \in V$, i.e. $VU_{max}(G,G-L,t)$. 

Note that our preliminary attempt at defining link privacy above does not accommodate links 
where either of its endpoints have degree 1, since removal of that link disconnects the graph. 
This is illustrated in Figure~\ref{fig:si-privacy}, which depicts the cumulative distribution of 
$\epsilon$ SI link privacy values. We can see that approximately $30\%$ and $20\%$ of links in 
the Facebook interaction graph and friendship graphs respectively do not receive any privacy protection 
under this definition (because they are connected to degree 1 vertices). 
For the remaining links, we can see a similar qualitative trend as before: 
increasing the perturbation parameter $t$ significantly improves link privacy. To overcome the 
limitations of this definition, we propose an alternate formulation based on the notion of 
link equivalence.

\begin{mydef}
We define the structural equivalence (SE) between a link L' with a link L in graph G with respect 
to a perturbation mechanism M, as the statistical distance between probability distributions 
of the output of the perturbation mechanism, when using (a) the original graph G as an input 
to the perturbation mechanism, and (b) the graph G - L + L' as the input to the perturbation 
mechanism. A link L has K-anonymous $\epsilon$ SE privacy, if there exist at least K links L', 
such that $||P(G' = M(G)) - P(G'=M(G-L+L'))|| < \epsilon$. 
\end{mydef}

Observe that this definition of privacy is able to account for degree 1 vertices, since they can become 
connected to the graph via the addition of other alternate links L'. For our experiments, we limited 
the number of alternate links explored to $1000$ links (for computational tractability). 
Figure~\ref{fig:se-privacy} depicts the cumulative distribution of anonymity set sizes for links 
using $\epsilon=0.1$ for the Facebook interaction and friendship graphs. For 
$t=2$ we see a very similar trend as in the previous definition, where a non-trivial fraction of links 
do not receive much privacy. Unlike the previous setting however, as we increase the 
perturbation parameter $t$, the anonymity set size for even these links improves significantly. 
Using $t=10$, $50\%$ and $70\%$ of the links in the interaction and friendship graphs respectively, 
achieved the maximum tested anonymity set size of $1000$ links.

Connection with differential privacy~\cite{differential-privacy}: there is an interesting 
connection between our risk based privacy definitions and differential privacy. A  
differentially private mechanism adaptively adds noise to the system to ensure that \emph{all} 
user records in a database (links in our setting) receive a threshold privacy protection. In 
our mechanism, we are adding a fixed amount of noise (governed by the perturbation parameter $t$), 
and observing the variance in $\epsilon$ and anonymity set size values. 

%% file: applications.tex
\section{Applications}
\label{sec:applications}

In this section, we demonstrate the applicability of our perturbation 
mechanism to social network based systems. 

\subsection{Secure Routing}

Several peer-to-peer systems perform routing over social graph to improve 
performance and security, such as Sprout~\cite{sprout},Tribler~\cite{tribler}, 
Whanau~\cite{whanau:nsdi10} and X-Vine~\cite{x-vine}. 
Next, we analyze the impact of our perturbation algorithm on the utility of Sprout. 

\subsubsection{Sprout}
Sprout is a routing system that enhances the security of conventional DHT routing by leveraging 
trusted social links when available. For example, when routing towards a DHT key, if leveraging 
a social link makes forward progress in the DHT namespace, then the social link is used for 
routing. The authors of Sprout considered a linear trust decay model, where a users' social 
contacts are trusted with probability $f$ (set to $0.95$ in ~\cite{sprout}), and the trust in 
other users decreases as a linear function of the shortest path distance between the users (a 
decrement of $0.05$ is used in ~\cite{sprout}). The decrement is bounded by a value that reflects 
the probability of a random user in the network being trusted (set to $0.6$ in ~\cite{sprout}).

The reliability of a DHT lookup in sprout is defined as the probability of all users in the path 
being trusted. Table~\ref{tbl:sprout} depicts the reliability of routing using a single DHT lookup in Sprout, 
for the original and the perturbed topologies. We used Chord as the underlying DHT system. 
For each perturbation parameter, our results were averaged over $100$ perturbed topologies. 
We can see that as the perturbation parameter increases, the utility of application decreases. 
For example, using the original Facebook interaction topology, the reliability of a single DHT path 
in sprout is $0.11$, which drops to $0.10$ and $0.096$ when using perturbed 
topologies with parameters $t=5$ and $t=10$ respectively. However, even when using $t=10$, the performance 
is better as compared with the scenario where social links are not used for routing (Chord's baseline 
performance of $0.075$). We can see similar results for the Facebook friendship graph as well.

\begin{table}[!t]
\caption{Path reliability using Sprout for a linear trust decay model.}
\label{tbl:sprout}
\footnotesize
\centering
\begin{tabular}{|l|c|c|c|c|l}
\hline
 \multicolumn{2}{|c|}{\bf Facebook interaction graph} & \multicolumn{2}{|c|}{\bf Facebook friendship graph} \\
        { Mechanism} & {Reliability } & { Mechanism} & {Reliability}  \\ \hline \hline
 Original & 0.110 & Original & 0.140  \\ \hline
 t=3 & 0.101 & t=3 & 0.126  \\ \hline
 t=5 & 0.101 & t=5 & 0.121 \\ \hline
 t=10 & 0.096 & t=10 & 0.118  \\ \hline
 Chord & 0.075 & Chord & 0.072  \\ \hline
\end{tabular}
\vspace{-0.2in}
\end{table}

\subsection{Sybil detection}

\begin{figure*}[htp]
\centering
\mbox{
\hspace{-0.2in}
\hspace{-0.12in}
\begin{tabular}{c}
\psfig{figure=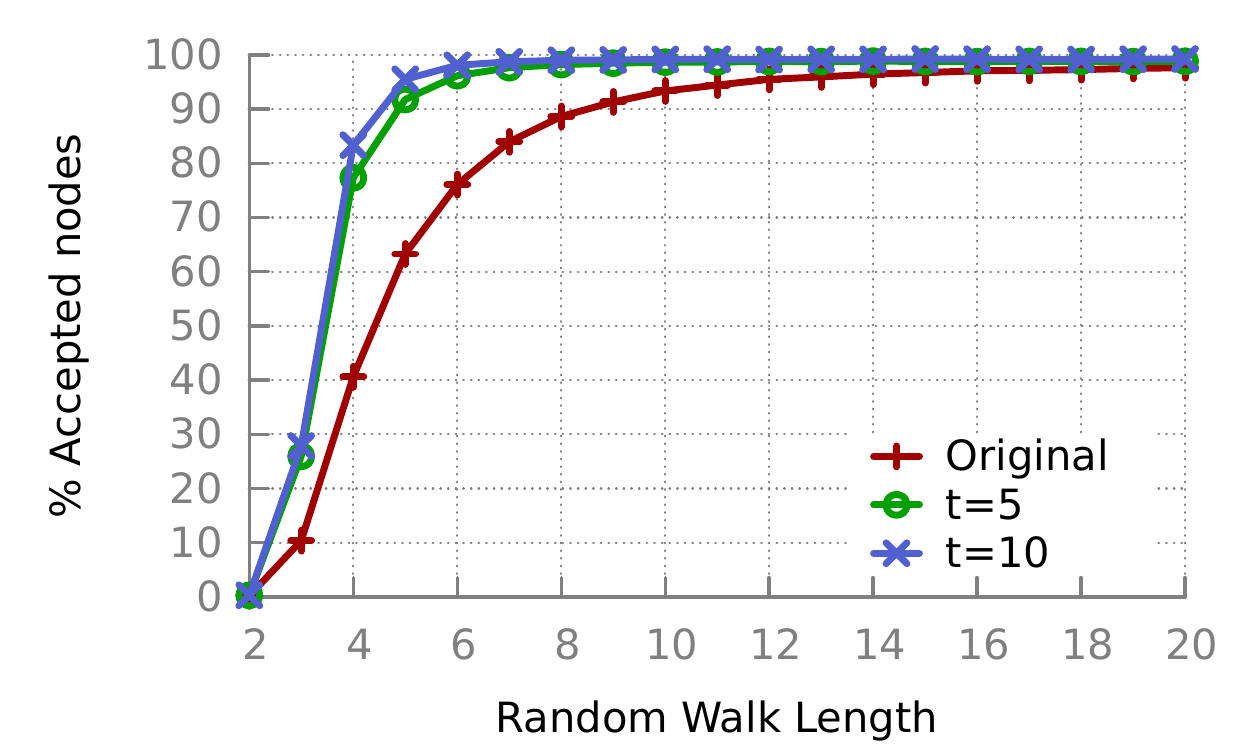,width=0.33 \textwidth}\vspace{-0.00in}\\
{(a)}
\end{tabular}
\begin{tabular}{c}
\psfig{figure=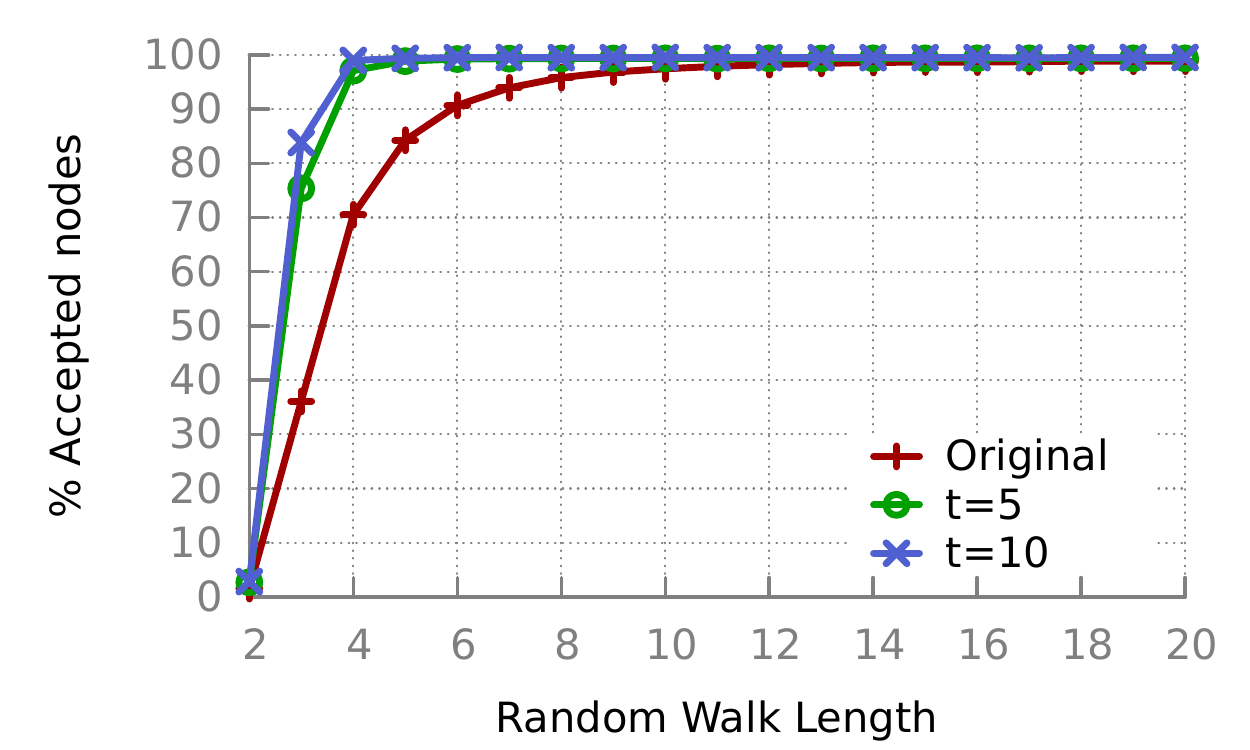,width=0.33 \textwidth}\vspace{-0.00in}\\
{(b)}
\end{tabular}
}
\vspace{-0.05in}
\caption{{\em SybilLimit \% validated honest nodes} as a function of SybilLimit random route length for (a) Facebook interaction graph (b) Facebook wall post graph. We can see that for a false positive rate 
of $1-2\%$, the required random route length for our perturbed topologies is a factor of $2-3$ smaller as compared with the original topology. Random route length is directly proportional to number 
of Sybil identities that can be inserted in the system.}
\label{fig:sybillimit-fp}
\vspace{-0.1in}
\end{figure*}

\begin{figure*}[htp]
\centering
\mbox{
\hspace{-0.2in}
\hspace{-0.12in}
\begin{tabular}{c}
\psfig{figure=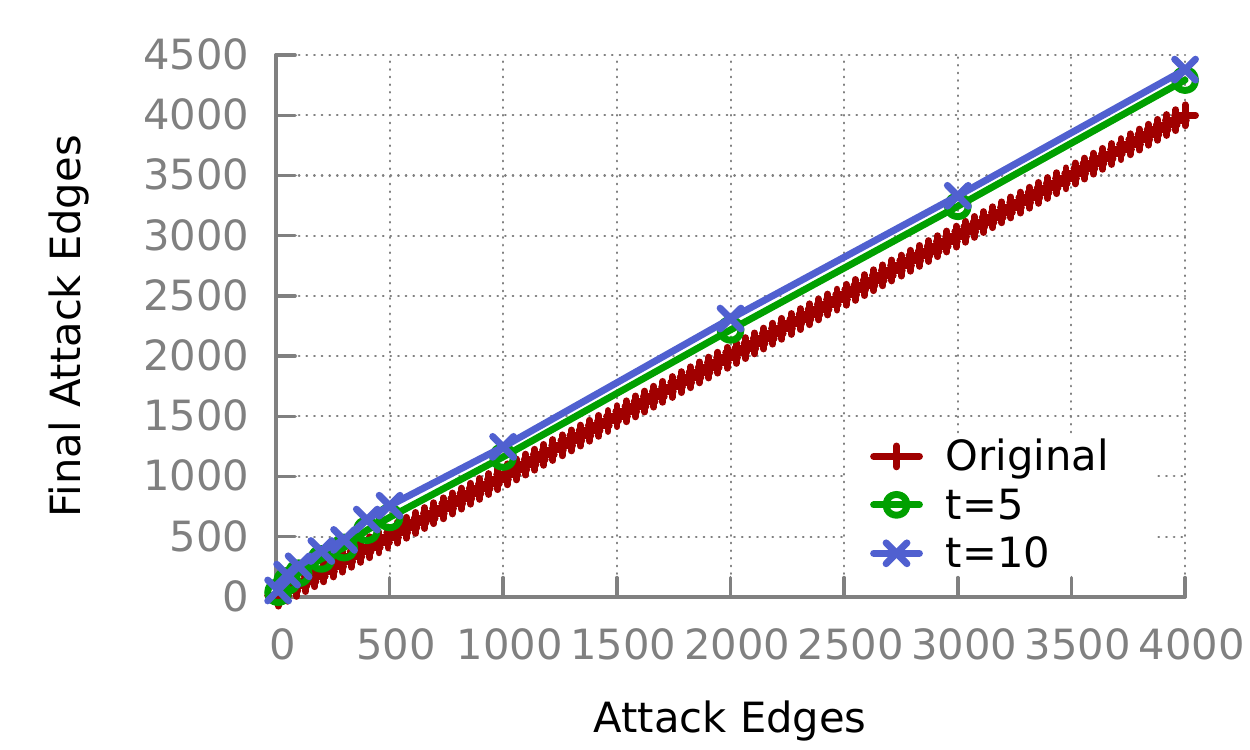,width=0.33 \textwidth}\vspace{-0.00in}\\
{(a)}
\end{tabular}
\begin{tabular}{c}
\psfig{figure=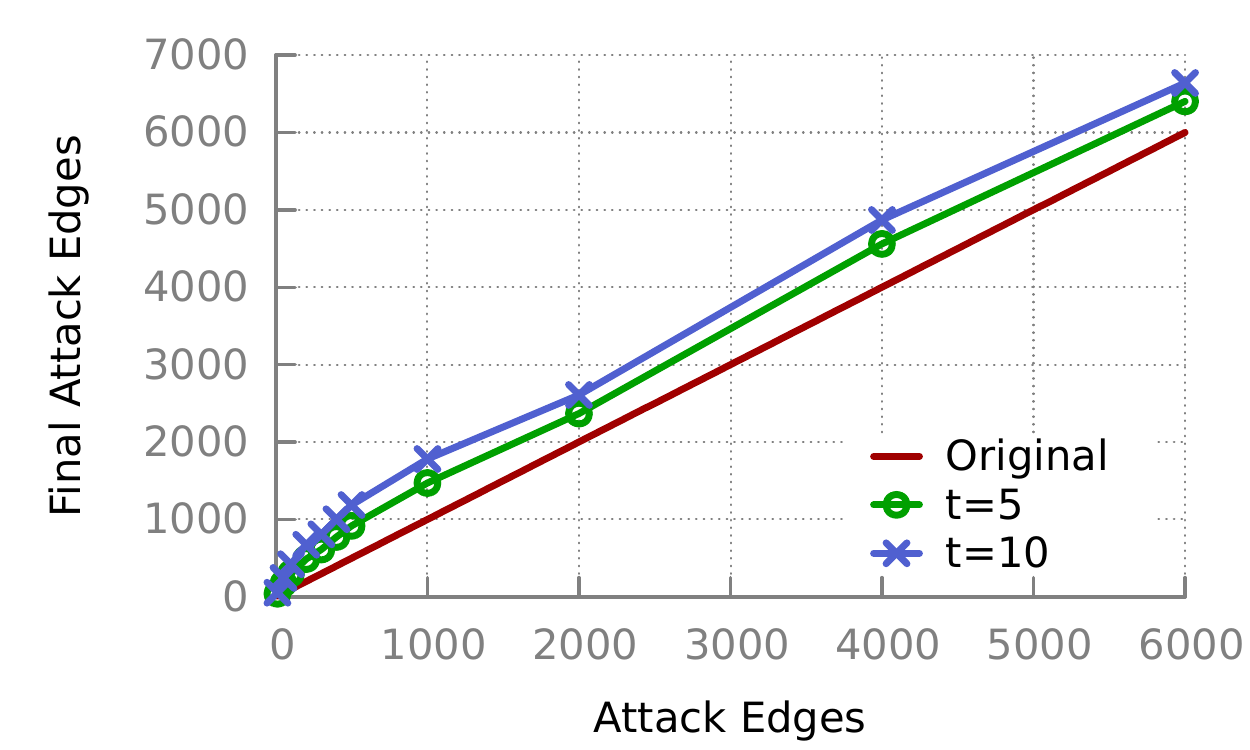,width=0.33 \textwidth}\vspace{-0.00in}\\
{(b)}
\end{tabular}
}
\vspace{-0.05in}
\caption{{\em Attack edges in perturbed topologies} as a function of attack edges in the original topology. We can see that there is a marginal increase in the number of attack edges in 
the perturbed topologies. The attack edges are directly proportional to the number of Sybil identities that can be inserted in the system.}
\label{fig:sybillimit-fn}
\vspace{-0.1in}
\end{figure*}

\begin{figure*}[htp]
\centering
\mbox{
\hspace{-0.2in}
\hspace{-0.12in}
\begin{tabular}{c}
\psfig{figure=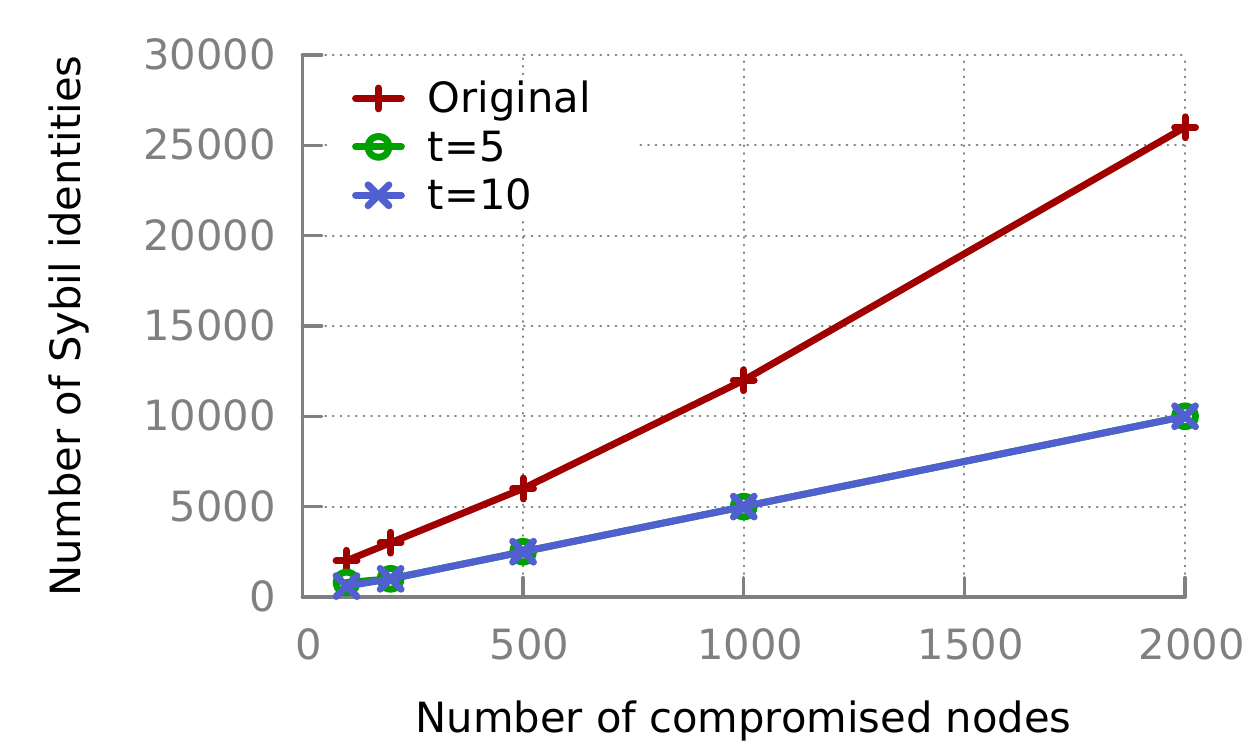,width=0.33 \textwidth}\vspace{-0.00in}\\
{(a)}
\end{tabular}
\begin{tabular}{c}
\psfig{figure=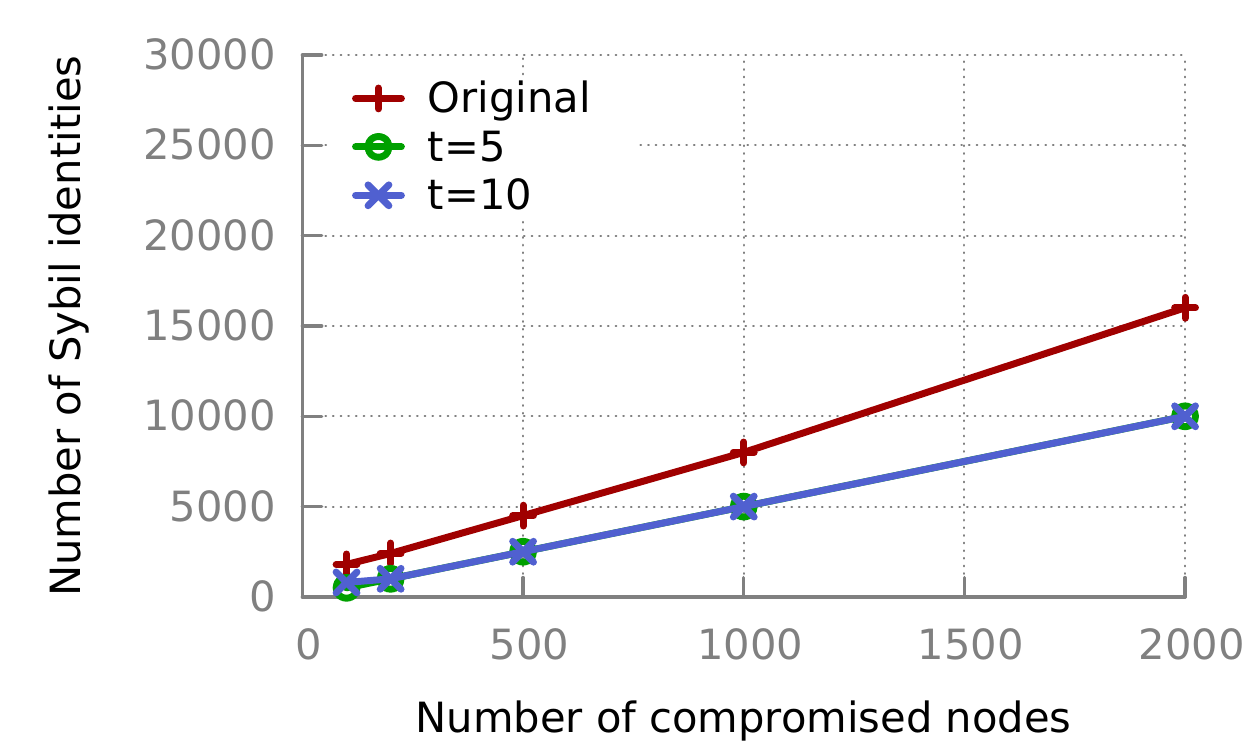,width=0.33 \textwidth}\vspace{-0.00in}\\
{(b)}
\end{tabular}
}
\vspace{-0.05in}
\caption{{\em Number of Sybil identities accepted by SybilInfer} as a function of number of compromised nodes 
in the original topology. We can see that there is a significant decline in the number of Sybil identities using 
our perturbation algorithms.}
\label{fig:sybilinfer}
\vspace{-0.1in}
\end{figure*}

In a Sybil attack~\cite{sybil}, a single user or an entity emulates the 
behavior of multiple identities in a system. Many systems 
are built on the assumption that there is a bound on the 
fraction of malicious nodes in the systems. By being able 
to insert a large number of malicious identities, 
an attacker can compromise the security properties of such 
systems. Sybil attacks are a powerful threat against both 
centralized as well as distributed systems, such as 
reputation systems, consensus and leader election protocols, 
peer-to-peer systems, anonymity systems, and recommendation 
systems. 

A wide body of recent work has proposed to leverage 
trust relationships embedded in social networks for 
detecting Sybil attacks~\cite{sybilguard,sybillimit,
sybilinfer,mohaisen:infocom11,tran:infocom11}. 
However, in all of these mechanisms, an adversary can 
learn the trust relationships in the social network. 
Next, we show that our perturbation algorithm preserves 
the ability of above mechanisms to detect Sybils, 
while protecting the privacy of the social network trust relationships. 

\subsubsection{SybilLimit}
We use SybilLimit~\cite{sybillimit} as a representative 
protocol for Sybil detection, since it is the most popular 
and well understood mechanism in the literature. SybilLimit 
is a probabilistic defense, and has both false positives 
(honest users misclassified as Sybils) and false negatives 
(Sybils misclassified) as honest users. 

We compared the performance of running SybilLimit on original 
graph, and on the transformed graph, for varying perturbation 
parameters. For each perturbation parameter, we averaged the 
results over $100$ perturbed topologies. Figure~\ref{fig:sybillimit-fp} 
depicts the percentage of honest users validated by SybilLimit using 
the original graph and perturbed graphs, as a function of 
the SybilLimit random route length (application parameter $w$). 
For any value of the SybilLimit random route length, the 
percentage of honest nodes accepted by SybilLimit is higher 
when using perturbed graphs. This is because our perturbation 
algorithms reduce the mixing time of the graph. In fact, for 
a false positive percentage of $1-2\%$ (99-98\% accepted nodes), 
the required length of the SybilLimit random routes is a factor 
of $2-3$ smaller as compared to the original topology. 
SybilLimit random routes are directly proportional to the number 
of Sybil identities that can be inserted in the system.  

This improvement in the number of accepted honest nodes (reduction in 
false positives) comes at the cost of increase in the number of 
attack edges between honest users and the attacker. Figure~\ref{fig:sybillimit-fn} 
depicts the number of attack edges in the perturbed topologies, for varying 
values of attack edges in the original graph. We can see that as expected, 
there is a marginal increase in the number of attack edges in the perturbed topologies.

Remark: The number of Sybil identities that an adversary can insert is given by $S=g' \cdot w'$. We note 
that the marginal increase in the number of attack edges $g'$ is offset by the reduced length of the 
SybilLimit random route parameter $w$ (for any desired false positive rate), thus achieving 
comparable performance with the original social graph. In fact, for perturbed topologies, since 
the required random route length in SybilLimit is halved for a false positive rate of $1-2\%$, 
and the increase in the number of attack edges is less than a factor of two, the Sybil defense 
performance has \emph{improved} using our perturbation mechanism. Thus for Sybil 
defenses, our perturbation mechanism is of independent interest, even without considering the 
benefit of link privacy. We further validate this conclusion using another state-of-art detection 
mechanism called SybilInfer.  

\subsubsection{SybilInfer}
We compared the performance of running SybilInfer on real and perturbed topologies. 
Figure~\ref{fig:sybilinfer} depicts the optimal number of Sybil identities that an 
adversary can insert before being detected by SybilInfer, as a function of real 
compromised and colluding users. Again, we can see that the performance of perturbed 
graphs is better than using original graphs. This is due to the interplay between 
mixing time of graphs and the number of attack edges in the Sybil defense application. Our perturbation 
mechanism significantly reduces the mixing time of the graphs, while suffering only a 
marginal increase in the number of attack edges. It is interesting to see that the 
advantage of using our perturbation mechanism is less in Figure~\ref{fig:sybilinfer}(b), 
as compared to Figure~\ref{fig:sybilinfer}(a). This is because the mixing time of the 
Facebook friendship graph is much better (as compared with the the mixing time of the 
Facebook interaction graph). Thus we conclude that our perturbation mechanism improves 
the overall Sybil detection performance of existing approaches, especially for interaction 
based topologies that exhibit relatively poor mixing characteristics.

%% file: conclusion.tex
\section{Conclusion and Future Work}

In this work, we proposed a random walk based perturbation algorithm 
that anonymizes \emph{links} in a social graph while preserving the 
graph theoretic properties of the original graph. We provided formal definitions for 
utility of a perturbed graph from the perspective of vertices, related 
our definitions to global graph properties, and empirically analyzed the 
properties of our perturbation algorithm using real world social 
networks. Furthermore, we analyzed the privacy of our perturbation mechanism from 
several perspectives (a) a Bayesian viewpoint that takes into consideration 
specific adversarial prior, and (b) a risk based view point that is independent 
of the adversary's prior. We also formalized the relationship between utility 
and privacy of perturbed graphs. Finally, we experimentally demonstrated the 
applicability of our techniques on applications such as Sybil defenses and 
secure routing. For Sybil defenses, we found that our techniques are of 
independent interest. 

Our work opens several directions for future research, including (a) investigating 
the applicability of our techniques on directed graphs (b) modeling closed form 
expressions for computing link privacy using the Bayesian framework, and (c) investigating 
tighter bounds on $\epsilon$ for computing link privacy in the risk based framework, 
and (d) modeling temporal dynamics of social networks in quantifying link privacy. 

By protecting the privacy of trust relationships, we believe that our perturbation 
mechanism can act as a key enabler for real world deployment of secure systems that 
leverage social links.

%% file: appendix.tex
\appendix
\label{sec:appendix}

\subsection{Proof of Theorem~\ref{thm:utility-mixing}: Relating vertex utility and mixing time}

We now sketch the proof of the above theorem. From the definition of total variation distance, we can see that:

\begin{align}
||P_v^t(G') - \pi||_{tvd} \leq ||P_v^t(G')-P_v^t(G)||_{tvd} + ||P_v^t(G)-\pi||_{tvd}
\end{align}

From the definition of mixing time, we have that $\forall t \geq T_G(\epsilon)$: 

\begin{align}
 ||P_v^t(G') - \pi||_{tvd} \leq ||P_v^t(G')-P_v^t(G)||_{tvd} + \epsilon
\end{align}

Substituting $t=\tau_G(\epsilon)$ in the above equation, and taking the maximum over all vertices, we have that: 

\begin{align}
 \max_v ||P_v^{\tau_G(\epsilon)}(G') - \pi||_{tvd} & \leq  \max_v ||P_v^{\tau_G(\epsilon)}(G')- \nonumber \\ P_v^{\tau_G(\epsilon)}(G)||_{tvd} + \epsilon \nonumber \\
 \max_v ||P_v^{\tau_G(\epsilon)}(G') - \pi||_{tvd} & \leq  VU_{max}(G,G',\tau_G(\epsilon)) + \epsilon \nonumber \\
\end{align}

Finally, we have that: 

\begin{align}
\tau_{G'}(\epsilon+VU_{max}(G,G',\tau_G(\epsilon)) \leq \tau_G(\epsilon)
\end{align}


\subsection{Proof of Theorem~\ref{thm:utility-slem}: Relating vertex utility and SLEM}

We now sketch the proof of the above theorem. It is known that for undirected graphs, the second largest 
eigenvalue modulus is related to the mixing time of the graph as follows~\cite{aldous}:

\begin{align}
\frac{\mu_{G'}}{2(1-\mu_{G'})} \log(\frac{1}{2\epsilon}) \leq \tau_{G'}(\epsilon) \leq \frac{\log n + \log(\frac{1}{\epsilon})}{1-\mu_{G'}}
\end{align}

From the above equation, we can bound the SLEM in terms of the mixing time as follows: 

\begin{align}
1-\frac{\log n + \log(\frac{1}{\epsilon})}{\tau_{G'}(\epsilon)} \leq \mu_{G'} \leq \frac{2\tau_{G'}(\epsilon)}{2\tau_{G'}(\epsilon) + \log(\frac{1}{2\epsilon})}	
\end{align}

Replacing $\epsilon$ with $\epsilon + VU_{max}(G,G',\tau_G(\epsilon))$, we have that:

\begin{align}
1-\frac{\log n + \log(\frac{1}{\epsilon+VU_{max}(G,G',\tau_G(\epsilon))})}{\tau_{G'}(\epsilon+VU_{max}(G,G',\tau_G(\epsilon))} \leq \mu_{G'} \leq \nonumber \\ 
\frac{2\tau_{G'}(\epsilon+VU_{max}(G,G',\tau_G(\epsilon))}{2\tau_{G'}(\epsilon+VU_{max}(G,G',\tau_G(\epsilon)) + \log(\frac{1}{2\epsilon+2VU_{max}(G,G',\tau_G(\epsilon)})}	
\end{align}

Finally, we leverage $\tau_{G'}(\epsilon+VU_{max}(G,G',\tau_G(\epsilon)) \leq \tau_G(\epsilon)$ in the above equation, to get: 

\begin{align}
1-\frac{\log n + \log(\frac{1}{\epsilon+VU_{max}(G,G',\tau_G(\epsilon))})}{\tau_{G}(\epsilon)} \leq \mu_{G'} \leq \nonumber \\ 
\frac{2\tau_{G}(\epsilon)}{2\tau_{G}(\epsilon) + \log(\frac{1}{2\epsilon+2VU_{max}(G,G',\tau_G(\epsilon)})}	
\end{align}

\subsection{Proof of Theorem~\ref{thm:mixing}: Bounding Mixing time}

Observe that the edges in graph G' can be modeled as samples from the 
$t$ hop probability distribution of random walks starting from vertices in 
G. We will prove the lowerbound on the mixing time of the perturbed
graph G' by contradiction: let us suppose that the mixing time of the 
graph G' is $k < \frac{\tau_G(\epsilon)}{t}$. Then in the original graph 
G, a user could have performed random walks of length $k \cdot t$ and 
achieve a variation distance less than $\epsilon$. But $k \cdot t < \tau_G(\epsilon)$, 
which is a contradiction. Thus, we have that $\frac{\tau_G(\epsilon)}{t} \leq \tau_{G'}(\epsilon)$.

We prove an upper bound on mixing time of the perturbed graph using the notion of
graph conductance. Let us denote the number of edges across the bottleneck cut (say $S$) of the 
original topology as $g$. Observe that the $t$ hop conductance between the sets $S$ and $\overline{S}$ 
is strictly larger than the corresponding one hop conductance (since $S$ is the bottleneck cut 
in the original topology). Thus, $E(G') \geq g$. Hence the expected graph conductance is an increasing 
function of the perturbation parameter $t$, and thus $E(\tau_{G'}(\epsilon)) \leq \tau_G(\epsilon)$.

\subsection{Proof of Theorem~\ref{thm:utility-privacy}: Relating utility and privacy}

From the definition of maximum vertex utility, we have that $|P_{AB}^l(G) - P_{AB}^l(G')|  \leq VU_{max}(G,G',l)$.                  

Thus, we can bound $P_{AB}^l(G)$ as follows: 
\begin{align}
 P_{AB}^l(G') - VU_{max}(G,G',l)  \leq  P_{AB}^l(G) \leq \nonumber \\ P_{AB}^l(G') + VU_{max}(G,G',l)  
\end{align} 

Thus for any value of $k$, if $P_{AB}^k(G') - VU_{max}(G,G',k) > 0$, then we have that the 
lower bound on the probability $P_{AB}^k(G) > 0$, which reveals the information that $A$ and 
$B$ are within an $k$ hop neighborhood of each other. Thus the maximum information is revealed 
when the value of $k$ is minimized, while maintaining $P_{AB}^k(G') - VU_{max}(G,G',k) > 0$, i.e., 
$k=\delta$. This gives us a lower bound on the probability of $A$ and $B$ being friends in the 
original graph: the prior probability that two vertices in a $\delta$ hop neighborhood are friends: 
$f(\delta)$. Let $m_\delta$ denote the average number of links in a $\delta$ hop neighborhood, 
and let $n_\delta$ denote the average number of vertices in a $\delta$ hop neighborhood. In the special 
case of a null prior, we have that $f(\delta) = m_{\delta}/\binom{n_{\delta}}{2}$.